\documentclass[a4paper, 12pt]{article}

\usepackage[utf8]{inputenc}
\usepackage[T1]{fontenc}
\usepackage[english]{babel}
\usepackage[affil-it]{authblk}
\usepackage{etex}
\usepackage{xcolor}
\usepackage{tabularx}
\usepackage{numprint}
\usepackage{mathrsfs}
\usepackage{mathtools}
\usepackage{amsmath}
\usepackage{amssymb}
\usepackage{cancel}
\usepackage{enumerate}
\usepackage{bussproofs}
\usepackage{proof}
\usepackage[ND,SEQ]{prftree}
\usepackage{qtree}
\usepackage[colorlinks,citecolor=red,urlcolor=blue,bookmarks=false,hypertexnames=true]{hyperref} 
\usepackage{pdflscape}
\usepackage{hyphenat}
\usepackage{rotating}
\usepackage{relsize}
\usepackage{natbib}
\usepackage{multirow}
\usepackage{dashbox}
\usepackage{geometry}
\usepackage{changepage}
\usepackage{xspace}
\usepackage{algorithm}
\usepackage{algorithmic}
\usepackage[disable,colorinlistoftodos,prependcaption,textsize=footnotesize]{todonotes}
\newcommand{\imply}{\ensuremath{\rightarrow}}

\newcommand{\ipl}{\ensuremath{\mathbf{Int}}\xspace}
\newcommand{\mpl}{\ensuremath{\mathbf{Min}}\xspace}
\newcommand{\mil}{\ensuremath{\mathbf{M}^{\rightarrow}}\xspace}

\newcommand{\lk}{\ensuremath{\mathbf{LK}}\xspace}
\newcommand{\lj}{\ensuremath{\mathbf{LJ}}\xspace}

\newcommand{\ljarrow}{\ensuremath{\mathbf{LJ^{\rightarrow}}}\xspace}

\newcommand{\lmt}{\ensuremath{\mathbf{LMT^{\rightarrow}}}\xspace}

\newcommand{\SEQ}{\ensuremath{\Rightarrow}}
\newcommand{\Seq}[2]{\ensuremath{#1 \SEQ #2}}

\newenvironment{dedsystem}
{
 \newcolumntype{Y}{>{\centering\arraybackslash}X}
 \setlength{\extrarowheight}{3ex}
 \tabularx{\hsize}{|Y|}}
{\endtabularx}

\newenvironment{blockquote}{%
\par%
\medskip
\leftskip=4em\rightskip=2em%
\noindent\ignorespaces}{%
\par\medskip}

\newlength{\overwritelength}
\newlength{\minimumoverwritelength}
\newcommand{\overwrite}[3]{%
  \settowidth{\overwritelength}{$#2$}%
  \ifdim\overwritelength<\minimumoverwritelength%
  \setlength{\overwritelength}{\minimumoverwritelength}\fi%
  \stackrel
  {%
    \begin{minipage}{\overwritelength}%
      \color{#1}\centering\small #3\\%
      \rule{1pt}{9pt}%
    \end{minipage}}
  {\colorbox{#1!50}{\color{black}$\displaystyle#2$}}}

\newtheorem{definition}{Definition}
\newtheorem{proposition}{Proposition}
\newtheorem{theorem}{Theorem}
\newtheorem{lemma}{Lemma}
\newtheorem{observation}{Observation}
\newtheorem{proof}{Proof}
\newtheorem{example}[definition]{Example}

\title{A Sequent Calculus Proof Search Procedure and Counter-model Generation
based on Natural Deduction Bounds}

\author[1,3]{Jefferson de Barros Santos%
          \thanks{\texttt{jefferson.santos@fgv.br}; Corresponding author}
       }

\author[2]{Bruno Lopes Vieira%
          \thanks{\texttt{bruno@ic.uff.br}}
       }

\author[3]{Edward Hermann Haeusler%
          \thanks{\texttt{hermann@inf.puc-rio.br}}
       }       

\affil[1]{Escola Brasileira de Administra\c c\~ao P\'ublica e de Empresas \\
       Funda\c c\~ao Get\'ulio Vargas \\
       Rio de Janeiro - RJ, Brazil
      }       
\affil[2]{Instituto de Computa\c c\~ao \\
       Universidade Federal Fluminense \\
       Niter\'oi - RJ, Brazil
      }
\affil[3]{Departamento de Inform\'atica \\
       Pontif\'\i cia Universidade Cat\'olica do Rio de Janeiro \\
       Rio de Janeiro - RJ, Brazil
      }

\begin{document}

\maketitle

\begin{abstract}
  In a previously published ENTCS paper~(\cite{santos2016}), we introduced a
  sequent calculus called \lmt for Minimal Implicational Propositional Logic
  (\mil). This calculus provides a proof search procedure for \mil that works
  in a bottom-up approach. We proved there that \lmt is sound and complete. We
  also suggested a strategy to guarantee termination of the proof search
  procedure. In this current paper, we refined this strategy and presented a
  new strategy for~\lmt termination. Considering this new strategy, we also
  provide a (new) completeness proof for the system, which improves the
  previous version. Besides that, we present explicit upper bounds on the
  proof search procedure, derived from this new strategy. We also provide a
  full soundness proof of the system.
\end{abstract}


%%%%%%%%%%%%%%%%%%%%%%%%%%%%%%%%%%%%%%%%%%%%%%%%%%%%%%%%%%%%%%%%%%%%%%%%%%%%
\section{Introduction}
\label{sec:intro}
%%%%%%%%%%%%%%%%%%%%%%%%%%%%%%%%%%%%%%%%%%%%%%%%%%%%%%%%%%%%%%%%%%%%%%%%%%%%

Propositional Minimal Implicational Logic (\mil) is the fragment of the
Propositional Minimal Logic (\mpl) containing only the logical connective
\imply.

The $TAUT$ problem for \mil is the general problem of deciding if a formula
$\alpha \in \mil$ is always true. $TAUT$ is a PSPACE-Complete problem as stated
by~\cite{statman1974}, who also shows that this logic polynomially simulates
Propositional Intuitionistic Logic. Statman's simulation can also be used to
simulate Propositional Classical Logic polynomially.

Furthermore,
\cite{haeusler2014a} shows that \mil can polynomially simulate not only
Propositional Classical and Intuitionistic Logic but also the full Propositional
Minimal Logic and any other decidable propositional logic with a Natural
Deduction system where the Subformula Principle holds~(see \cite{prawitz1965}).

Moreover, \mil has a strong relation with open questions about the Computational
Complexity Hierarchy, as we can see from the statements below.

\begin{itemize}
\item If $CoNP \neq NP$ then $NP \neq P$.

\item If $PSPACE = NP$ then $CoNP = NP$.
  
\item $CoNP=NP$, iff $\exists DS$, a deductive
  system, such that $\forall\alpha\in$ $TAUT_{Cla}$, there is $\Pi$, a proof of $\alpha$ in $DS$, $size(\Pi_{DS}) \leq
  Poly(|\alpha|)$ and the fact that $\Pi$ is a proof of $\alpha$ is verifiable
  in polynomial time on the size of $\alpha$.

\item $PSPACE=NP$ iff $\exists DS$, a deductive
  system, such that $\forall\alpha\in$ $TAUT_{\mil}$, there is $\Pi$, a proof of $\alpha$ in $DS$, $size(\Pi_{DS}) \leq
  Poly(|\alpha|)$ and the fact that $\Pi$ is a proof of $\alpha$ is verifiable
  in polynomial time on the size of $\alpha$.
\end{itemize}

Those characteristics show us that \mil is as hard to implement as the most
popular propositional logics. This fact, together with its straightforward language
(only one logical connective), makes \mil a research object that can
provide us with many insights about the complexity class relationships
mentioned above and about the complexity of many other logics. Moreover, the
problem of conducting a proof search in a deductive system for~\mil has the
complexity of $TAUT$ as a lower bound. Thus, the size of propositional proofs
may be huge, and \emph{automated theorem provers} (ATP) should take care of
super-polynomially sized proofs. Therefore, the study of deductive systems
for~\mil can directly influence in techniques to improve the way provers manage
such proofs.

In \cite{santos2016}, we presented a sound and complete \emph{sequent calculus}
for \mil. We named it \lmt. This calculus establishes a bottom-up approach for
proof search in \mil using a unified procedure either for provability and
counter-model generation. \lmt was designed to avoid the usage of \emph{loop
  checkers} and mechanisms for \emph{backtracking} in its implementation.
Counter-model generation (using Kripke semantics) is achieved as a consequence
of the way the tree (produced by a failed proof search) is
constructed during a proof search process.

In this current work, we present an upper bound to the proof search procedure of
\lmt via translation functions from very known deductive systems for \mil,
Prawitz's Natural Deduction (\cite{prawitz1965}) and Gentzen's sequent calculus
(\cite{gentzen1935}). These translation functions together with a strategy to
apply the rules of \lmt~provide termination for the proof search procedure.
Besides that, we show here all counter-model generation cases, including those
missed in \cite{santos2016}.

We start in the next section with a brief discussion on the syntax and
semantics of \mil used through this text. Section 3 discusses some related work
and state of the art in the field of deductive systems for \mil. In Section
4 we present a study about the size of proofs in \mil~ to establish a bound for
proof search that can be used as a limit in the termination procedure of \lmt.
Section 5 presents \lmt~itself and its main features: termination, soundness,
and completeness (with the counter-model generation as a corollary). Section 6
concludes the paper discussing some open problems and future work.

%%%%%%%%%%%%%%%%%%%%%%%%%%%%%%%%%%%%%%%%%%%%%%%%%%%%%%%%%%%%%%%%%%%%%%%%%%%%
\section{Minimal Implicational Logic}
\label{sec:mimp}
%%%%%%%%%%%%%%%%%%%%%%%%%%%%%%%%%%%%%%%%%%%%%%%%%%%%%%%%%%%%%%%%%%%%%%%%%%%%

We can formally define the language $\mathcal{L}$ for \mil as follows. 

\begin{definition}
  The \emph{alphabet} of $\mathcal{L}$ consists of:
  \begin{itemize}
  \item An enumerable set of propositional symbols, called atoms.
  \item The binary connective (or logical operator) for implication ($\imply$).
  \item parentheses: ``$($'' and ``$)$''.
  \end{itemize}
\end{definition}

\begin{definition}
  We can define the general notion of a \emph{formula} in $\mathcal{L}$
  inductively:
  \begin{itemize}
  \item Every propositional symbol is a formula in $\mathcal{L}$. We call them
    atomic formulas.
  \item If $A$ and $B$ are formulas in $\mathcal{L}$ then $(A \imply B)$ are
    also.
  \end{itemize}
\end{definition}

As usual, parentheses are used for disambiguation. We use the following
conventions through the text:
\begin{itemize}
\item Upper case letters to represent atomic formulas: $A$, $B$, $C$, \ldots
\item Lower Greek letters to represent generic formulas: $\alpha$, $\beta$,
  \ldots
\item Upper case Greek letters are used to represent sets of formulas. For
  example $\Delta$, $\Gamma$, \ldots
\item If parentheses are omitted, implications are interpreted right nested.
\end{itemize}

The semantics of \mil is the intuitionistic semantics restricted to
\imply~on\-ly. Thus, given a propositional language $\mathcal{L}$, a \mil model
is a structure $\left< U,\preceq, \mathcal{V} \right>$, where $U$ is a non-empty
set (worlds), $\preceq$ is a partial order relation on $U$ and $\mathcal{V}$ is
a function from $U$ into the power set of $\mathcal{L}$, such that if $i,j\in U$
and $i\preceq j$ then $\mathcal{V}(i)\subseteq\mathcal{V}(j)$. Given a model,
the satisfaction relationship $\models$ between worlds in models and formulas is
defined as in Intuitionistic Logic, namely:
\begin{itemize}
\item $\left< U,\preceq, \mathcal{V} \right>\models_{i}p$, $p\in\mathcal{L}$,
  iff, $p\in\mathcal{V}(i)$
\item $\left< U,\preceq, \mathcal{V} \right>\models_{i}\alpha_1\imply\alpha_2$,
  iff, for every $j\in U$, such that $i\preceq j$, if $\left< U,\preceq,
    \mathcal{V} \right>\models_{j}\alpha_1$ then $\left< U,\preceq, \mathcal{V}
  \right>\models_{j}\alpha_2$.
\end{itemize}

As usual a formula $\alpha$ is valid in a model $\mathcal{M}$, namely
$\mathcal{M}\models\alpha$, if and only if, it is satisfiable in every world $i$
of the model, namely $\forall i\in U, \mathcal{M}\models_{i}\alpha$. A formula
is a \mil tautology, if and only if, it is valid in every model.

%%%%%%%%%%%%%%%%%%%%%%%%%%%%%%%%%%%%%%%%%%%%%%%%%%%%%%%%%%%%%%%%%%%%%%%%%%%%
\section{Related Work}
\label{sec:related-work}
%%%%%%%%%%%%%%%%%%%%%%%%%%%%%%%%%%%%%%%%%%%%%%%%%%%%%%%%%%%%%%%%%%%%%%%%%%%%

It is known that Prawitz's Natural Deduction System for Propositional Minimal
Logic with only the \imply-rules (\imply-Elim and \imply-Intro) is sound and
complete for the \mil regarding Kripke semantics. As a consequence of this,
Gentzen's \lj system (\cite{gentzen1935}) containing only right and left
\imply-rules is also sound and complete.

In~\cite{gentzen1935}~, Gentzen proved the decidability of the Propositional
Intuitionistic Logic (\ipl), which also includes the cases for \mpl and \mil.

However, Gentzen's approach was not conceived to be a bottom-up proof search
procedure. Figure~\ref{fig:gentzen} shows structural and logic rules of an
adapted Gentzen's sequent calculus for \mil, called \ljarrow. We restrict the
right side of a sequent to one and only one formula (we are in \mil; thus,
sequents with an empty right side does not make sense). This restriction implies
that structural rules can only be considered for main formulas on the left side
of a sequent. \ljarrow incorporates contraction in the \imply-left, with the
repetition of the main formula of the conclusion on the premises. We use those
adaptations to explain the difficulties in using sequent calculus systems for
proof search in \ipl, \mpl, and \mil. In~\cite{dyckhoff2016}, Dyckhoff describes
in detail the evolution of those adaptations over Gentzen's original~\lj system.
Those adaptations are attempting to improve bottom-up proof search mechanisms
for \ipl.

A central aspect when considering mechanisms for proof search in \mil (and also
for \ipl) is the application of the \imply-left rule. The \lk system proposed by
\cite{gentzen1935}, the sequent calculus for Classical Logic, with some
adaptations (e.g., \cite{seldin1998}) can ensure that each rule reduces the
degree (the number of atomic symbols occurrences and connectives in a formula)
of the main formula of the sequent (the formula to which the rule is applied),
when applied in a bottom-up manner during the proof search. This fact implies
the termination of the system. However, the case for \ipl is more complicated.
First, we have the ``context-splitting'' (using an expression from
\cite{dyckhoff2016}) nature of \imply-left, i.e., the formula on the right side
of the conclusion sequent is lost in the left premise of the rule application.
Second, as we can reuse a hypothesis in different parts of a proof, the main
formula of the conclusion must be available to be used again by the generated
premises. Thus, the \imply-left rule has the repetition of the main formula in
the premises, a scenario that allows the occurrence of loops in automatic
procedures.

\begin{figure}[htb]
\begin{dedsystem}
\hline
  \AxiomC{}
  \RightLabel{axiom}
  \UnaryInfC{\Seq{\Delta, \gamma}{\gamma}}
\DisplayProof \\
  \AxiomC{\Seq{\Delta}{\gamma}}
  \RightLabel{weakening (w)}
  \UnaryInfC{\Seq{\alpha, \Delta}{\gamma}}
\DisplayProof \hspace{3em}
  \AxiomC{\Seq{\alpha, \alpha, \Delta}{\gamma}}
  \RightLabel{contraction (c)}
  \UnaryInfC{\Seq{\alpha, \Delta}{\gamma}}
\DisplayProof \\
  \AxiomC{\Seq{\Gamma, \alpha, \beta, \Delta}{\gamma}}
  \RightLabel{exchange (e)}
  \UnaryInfC{\Seq{\Gamma, \beta, \alpha, \Delta}{\gamma}}
\DisplayProof \hspace{3em}
  \AxiomC{\Seq{\Delta}{\alpha}}
  \AxiomC{\Seq{\alpha, \Gamma}{\gamma}}
  \RightLabel{cut}
  \BinaryInfC{\Seq{\Delta, \Gamma}{\gamma}}
\DisplayProof \\
  \AxiomC{\Seq{\Delta, \alpha}{\beta}}
  \RightLabel{$\imply$-right ($\imply$-r)}
  \UnaryInfC{\Seq{\Delta}{\alpha \imply \beta}}
\DisplayProof \\
  \AxiomC{\Seq{\Delta, \alpha \imply \beta}{\alpha}}
  \AxiomC{\Seq{\Delta, \alpha \imply \beta, \beta}{\gamma}}
  \RightLabel{$\imply$-left ($\imply$-l)}
  \BinaryInfC{\Seq{\Delta, \alpha \imply \beta}{\gamma}}
  \DisplayProof \\ [3ex]\hline
\end{dedsystem}
\caption{Rules of Gentzen's \lj}
\label{fig:gentzen}
\end{figure}

Since Gentzen, many others have explored solutions to deal with the challenges
mentioned above, proposing new calculi (sets of rules), strategies and proof
search procedures to allow more automated treatment to the problem.

Unfortunately, the majority of these results are focused on \ipl, with very few
works dedicated explicitly to \mpl or \mil. Thus, we needed to concentrate our
literature review in the \ipl case, adjusting the found results to the
\mil context by ourselves.

A crucial source of information was the work of \cite{dyckhoff2016}, appropriately
entitled \emph{``Intuitionistic decision procedures since Gentzen''} that
summarizes in chronological order the main results of this field. In the next paragraphs, we
highlight the most important of those results presented in \cite{dyckhoff2016}.

A common way to control the proof search procedure in \mil (and in \ipl) is by
the definition of routines for loop verification as proposed
in~\cite{underwood1990}. Dyckhoff present loop checkers as very expensive
procedures for automatic reasoning (\cite{dyckhoff2016}), although they are
effective to guarantee termination of proof search procedures. The work in
\cite{heuerding1996} and \cite{howe1997} are examples of techniques that can be
used to minimize the performance problems that can arise with the usage of such
procedures.

To avoid the use of loop checkers, \cite{dyckhoff1992} proposed a terminating
contraction-free sequent calculus for \ipl, named \ensuremath{\mathbf{LJT}},
using a technique based on the work of ~\cite{vorobev1970} in the 50s.
\cite{pinto1995} extended this work showing a method to generate
counter-examples in this system. They proposed two calculi, one for proof search
and another for counter-model generation, forming a way to decide about the
validity or not of formulas in \ipl. A characteristic of their systems is that
the subformula property does not hold on them. In \cite{ferrari2013}, a similar
approach is presented using systems where the subformula property holds. They
also proposed a single decision procedure for \ipl, which guarantees minimal depth
counter-model.

Focused sequent calculi appeared initially in the Andreoli's work on linear
logic \cite{andreoli1992}. The author identified a subset of proofs from
Gentzen-style sequent calculus, which is complete and tractable.
\cite{liang2007} proposed the focused sequent calculi \ensuremath{\mathbf{LJF}}
where they used a mapping of \ipl into linear logic and adapted the Andreoli's
system to work with the image. \cite{dyckhoff2006} presented the focused system
\ensuremath{\mathbf{LJQ}} that work direct in \ipl. Focusing is used in their
system as a way to implement restrictions in the \imply-left rule as proposed by
\cite{vorobev1970} and \cite{hudelmaier1993}. The work of \cite{dyckhoff2006}
follows from the calculus with the same name presented in~\cite{herbelin1995}.

\cite{dyckhoff2016} also identify a list of features particularly of interest
when evaluating mechanisms for proof search in \ipl that we will follow when
comparing our solution to the other existent ones. They are:
\textbf{termination} (proof search procedure stops both for theorems and
non-theorem formulas), \textbf{bicompleteness} (extractability of models from
failed proof searches), \textbf{avoidance of backtracking} (backtracking being a
very immediate approach to deal with the context split in \imply-left, but it is
also a complex procedure to implement), \textbf{simplicity} (allows easier
reasoning about systems).

%%%%%%%%%%%%%%%%%%%%%%%%%%%%%%%%%%%%%%%%%%%%%%%%%%%%%%%%%%%%%%%%%%%%%%%%%%%%
\section{The Size of Proofs in \mil} 
\label{sec:proofsize}
%%%%%%%%%%%%%%%%%%%%%%%%%%%%%%%%%%%%%%%%%%%%%%%%%%%%%%%%%%%%%%%%%%%%%%%%%%%%

Hirokawa has presented an upper bound for the size of normal form Natural
Deduction proofs of implicational formulas in~\ipl~ (that correspond to
\mil~formulas). \cite{Hirokawa1991} showed that for a formula $\alpha \in
\mil$, this limit is $|\alpha| \cdot 2^{|\alpha| + 1}$. As the Hirokawa result
concerns normal proofs in Natural Deduction we present now a translation of this
system to a cut-free sequent calculus, following the \ljarrow~rules presented in
Section~\ref{sec:related-work}, thus we can establish the limit for proof search
in \ljarrow~too.

Figure~\ref{fig:nd2lj} presents a recursively defined function\footnote{We use a
  semicolon to separate arguments of functions (in function definitions and
  function calls) instead of the most common approach to using commas. This
  change in convention aims to avoid confusion with the commas used to separate
  formulas and sets of formulas in sequent notation.} to translate Natural
Deduction normal proofs of \mil~formulas into \ljarrow~proofs (in a version of
the system without the cut rule).

%%%%%%%%%%%%%%%%%%%%%%%%%%%%%%%%%%%%%%
\begin{figure}[htb]
\begin{dedsystem}
\hline
{\small
%--------------------------
\textbf{Axioms:}
%--------------------------
  
\begin{equation*}
  F(\alpha; \Gamma) = \Seq{\Gamma, \alpha}{\alpha}
\end{equation*}

\bigskip

%--------------------------
\textbf{Case of \imply Introduction:}
%--------------------------
\begin{eqnarray*}
\begin{minipage}{0.5\textwidth}
  \begin{equation*}
  F\left(
    \alwaysNoLine
    \AxiomC{[$\alpha]^1$}
    \def\extraVskip{2pt}
    \UnaryInfC{$\prod$}
    \alwaysSingleLine
    \UnaryInfC{$\beta$}
    \RightLabel{$\imply$-I$^1$}
    \UnaryInfC{$\alpha\imply\beta$}
    \DisplayProof;
    \Gamma
  \right)
  \end{equation*}
\end{minipage} & \hspace{-.13\textwidth} = &
\begin{minipage}{0.5\textwidth}  
  \prftree[r]{$\imply\mathrm{r}$}
                {F\left(
                   \alwaysNoLine
                   \AxiomC{$\alpha$}
                   \def\extraVskip{2pt}
                   \UnaryInfC{$\prod$}
                   \alwaysSingleLine
                   \UnaryInfC{$\beta$}
                   \DisplayProof;
                   \{\alpha\} \cup \Gamma
                \right)}
                {\Seq{\Gamma}{\alpha\imply\beta}}
\end{minipage}
\end{eqnarray*}

\bigskip

%--------------------------
\textbf{Case of \imply Elimination:}
%--------------------------
\begin{equation*}
  F\left(
    \alwaysNoLine
    \def\extraVskip{2pt}
    \AxiomC{$\prod_1$}
    \UnaryInfC{$\alpha$}
    \alwaysSingleLine
              \AxiomC{$\alpha\imply\beta$}
    \BinaryInfC{$\beta$}
    \alwaysNoLine
    \UnaryInfC{$\prod_2$}
    \UnaryInfC{$C$}    
    \DisplayProof;
    \Gamma
  \right) =
  \end{equation*}
}

{\footnotesize
\begin{minipage}{1\textwidth}\vspace{.2cm}
  \prfinterspace=5em             
  \prftree[r]{$\imply\mathrm{l}$}
             {F\left(
               \alwaysNoLine
               \AxiomC{$\prod_1$}
               \def\extraVskip{2pt}
               \UnaryInfC{$\alpha$}
               \DisplayProof;
               \{\alpha \imply \beta\} \cup \Gamma
             \right)}
             {F\left(
               \alwaysNoLine
               \AxiomC{$\beta$}
               \def\extraVskip{2pt}
               \UnaryInfC{$\prod_2$}
               \UnaryInfC{$C$}
               \DisplayProof;
               \{\alpha \imply \beta\} \cup \Gamma
             \right)}  
            {
             \imply-lc\left(
                                  c\left(
                                  F\left(
                                     \alwaysNoLine
                                     \AxiomC{$\prod_1$}
                                     \def\extraVskip{2pt}
                                     \UnaryInfC{$\alpha$}
                                     \DisplayProof;
                                     \{\alpha \imply \beta\} \cup \Gamma
                                     \right)
                                  \right);
                                  c\left(
                                  F\left(
                                     \alwaysNoLine
                                     \AxiomC{$\beta$}
                                     \def\extraVskip{2pt}
                                     \UnaryInfC{$\prod_2$}
                                     \UnaryInfC{$C$}
                                     \DisplayProof;
                                     \{\alpha \imply \beta\} \cup \Gamma
                                     \right)
                                  \right);
                                  \alpha\imply\beta
                                \right)}
\end{minipage}
}

\\ [3ex]\hline
\end{dedsystem}
\caption{A recursively defined function to translate Natural Deduction proofs into \ljarrow}
\label{fig:nd2lj}
\end{figure}
%%%%%%%%%%%%%%%%%%%%%%%%%%%%%%%%%%%%%%

In this definition, $c$ is a function that returns the conclusion (last sequent)
of a \ljarrow~demonstration as showed in~\eqref{eq:c-function}. Also,
$\imply-lc$ is a function that receives two \ljarrow~sequents and a formula to
construct the conclusion of a $\imply-left$ rule application, as defined
in~\eqref{eq:imply-lc-function}.

\bigskip\bigskip

\begin{equation}
\label{eq:c-function}
   c\left(
      \alwaysNoLine
      \AxiomC{$\prod$}
      \def\extraVskip{2pt}
      \UnaryInfC{$\Seq{\Gamma}{\gamma}$}
      \DisplayProof
  \right) = \Seq{\Gamma}{\gamma}
\end{equation}

\bigskip

\begin{equation}
\label{eq:imply-lc-function}
   \imply-lc(\Seq{\Gamma}{\alpha}; \Seq{\beta, \Gamma}{\gamma}; \alpha\imply\beta) = \Seq{\Gamma, \alpha\imply\beta}{\gamma}
 \end{equation}

\newpage
 
\begin{proposition}
  Let
  \\
    \alwaysNoLine
    \AxiomC{$\Gamma$}
    \UnaryInfC{$\prod$}
    \UnaryInfC{$\alpha$}
    \DisplayProof, a normal Natural Deduction derivation and $\Gamma$, the set
    of undischarged formulas in $\Gamma$ then,

    $$
    F\left(
      \alwaysNoLine
      \AxiomC{$\Gamma$}
      \UnaryInfC{$\prod$}
      \UnaryInfC{$\alpha$}
      \DisplayProof;
      \emptyset
    \right) =
      \alwaysNoLine
      \AxiomC{$\prod'$}
      \UnaryInfC{$\Seq{\Gamma}{\alpha}$}
      \DisplayProof
    $$
    is a proof in \ljarrow.  
\end{proposition}

\begin{proof}
    By induction in the size of $\prod$.
\end{proof}

As an example of the translation produced by the function of
Figure~\ref{fig:nd2lj}, we show below each step of the translation of a Natural
Deduction proof (\ref{proof:nd}) into an \ljarrow~proof (\ref{proof:lj}). To
shorten the size of the proofs we collapsed repeated occurrences of formulas
when passed as the hypothesis argument of the recursive function call.

%%%%%%%%%%%%%%%%%%%%%%%%%%%%%%%%%%%%%%
\vspace{3em}

{\footnotesize
\begin{equation}
\label{proof:nd}    
      \alwaysSingleLine
      \def\extraVskip{2pt}
      \AxiomC{$[B]^2$}
                    \AxiomC{$[A]^1$}      
                                       \AxiomC{$[A \imply (B \imply C)]^3$}       
                    \RightLabel{$\imply$-E}      
                    \BinaryInfC{$B \imply C$}      
      \RightLabel{$\imply$-E}            
      \BinaryInfC{$C$}      
      \RightLabel{$\imply$-I$^1$}            
      \UnaryInfC{$A \imply C$}      
      \RightLabel{$\imply$-I$^2$}      
      \UnaryInfC{$B \imply (A \imply C)$}
      \RightLabel{$\imply$-I$^3$}
      \UnaryInfC{$(A \imply (B \imply C)) \imply (B \imply (A \imply C))$}
      \DisplayProof        
\end{equation}

\begin{center}
$\mathlarger{\mathlarger{\triangledown}}$  
\end{center}

\begin{displaymath}
  F\left(
      \alwaysNoLine
      \AxiomC{$[A \imply (B \imply C)]^3$}
      \def\extraVskip{2pt}
      \UnaryInfC{$\prod_1$}
      \UnaryInfC{$B \imply (A \imply C)$}
      \alwaysSingleLine
      \RightLabel{$\imply$-I$^3$}
      \UnaryInfC{$(A \imply (B \imply C)) \imply (B \imply (A \imply C))$}
      \DisplayProof;
      \emptyset
  \right)
\end{displaymath}

\begin{center}
$\mathlarger{\mathlarger{\triangledown}}$  
\end{center}

\begin{displaymath}
  \prftree[r]{$\imply\mathrm{r}$}
                {F\left(
                      \alwaysNoLine
                      \AxiomC{$[B]^2, A \imply (B \imply C)$}
                      \def\extraVskip{2pt}
                      \UnaryInfC{$\prod_2$}
                      \UnaryInfC{$A \imply C$}
                      \alwaysSingleLine
                      \RightLabel{$\imply$-I$^2$}
                      \UnaryInfC{$B \imply (A \imply C)$}
                      \DisplayProof;
                      \{A \imply (B \imply C)\}
                  \right)}
                {\Seq{}{(A \imply (B \imply C)) \imply (B \imply (A \imply C))}}       
\end{displaymath}

\begin{center}
$\mathlarger{\mathlarger{\triangledown}}$  
\end{center}

\begin{displaymath}
  \prftree[r]{$\imply\mathrm{r}$}
  {
        \prftree[r]{$\imply\mathrm{r}$}
                {F\left(
                      \alwaysNoLine
                      \AxiomC{$[A]^1, B, A \imply (B \imply C)$}
                      \def\extraVskip{2pt}
                      \UnaryInfC{$\prod_3$}
                      \UnaryInfC{$C$}
                      \alwaysSingleLine
                      \RightLabel{$\imply$-I$^1$}
                      \UnaryInfC{$A \imply C$}
                      \DisplayProof;
                      \{B, A \imply (B \imply C)\}
                  \right)}
                {\Seq{A \imply (B \imply C)}{B \imply (A \imply C)}}
  }
  {\Seq{}{(A \imply (B \imply C)) \imply (B \imply (A \imply C))}}       
\end{displaymath}

\begin{center}
$\mathlarger{\mathlarger{\triangledown}}$  
\end{center}

\begin{displaymath}
  \prftree[r]{$\imply\mathrm{r}$}
  {  
    \prftree[r]{$\imply\mathrm{r}$}
    {
        \prftree[r]{$\imply\mathrm{r}$}
                {F\left(
                      \alwaysSingleLine
                      \AxiomC{$B$}
                              \AxiomC{$A$}     \AxiomC{$A \imply (B \imply C)$}
                                  \RightLabel{$\imply$-E} 
                                  \BinaryInfC{$B \imply C$}
                      \RightLabel{$\imply$-E}
                      \BinaryInfC{$C$}
                      \DisplayProof;
                      \{A, B, A \imply (B \imply C)\}
                  \right)}
                {\Seq{B, A \imply (B \imply C)}{A \imply C}}
    }
    {\Seq{A \imply (B \imply C)}{B \imply (A \imply C)}}
  }
  {\Seq{}{(A \imply (B \imply C)) \imply (B \imply (A \imply C))}}       
\end{displaymath}
}

In the following steps consider that $\Gamma = \{A, B, A \imply (B \imply C)\}$.

{\footnotesize
\begin{center}
$\mathlarger{\mathlarger{\triangledown}}$  
\end{center}

\begin{displaymath}
  \prftree[r]{$\imply\mathrm{r}$}
  {  
    \prftree[r]{$\imply\mathrm{r}$}
    {
      \prftree[r]{$\imply\mathrm{r}$}
            {
             \prfinterspace=8em
             \prftree[r]{$\imply\mathrm{l}$}
                {F(A; \Gamma)}
                {F\left(
                      \alwaysSingleLine
                      \AxiomC{$B$}
                                  \AxiomC{$B \imply C$}
                      \RightLabel{$\imply$-E}
                      \BinaryInfC{$C$}
                      \DisplayProof;
                      \Gamma
                   \right)}
                {\imply-lc\left(
                                  c\left(
                                  F(A; \Gamma)
                                  \right);
                                  c\left(
                                   F\left(
                                         \alwaysSingleLine
                                         \AxiomC{$B$}
                                                \AxiomC{$B \imply C$}
                                         \RightLabel{$\imply$-E}
                                         \BinaryInfC{$C$}
                                         \DisplayProof;
                                         \Gamma
                                     \right)
                                  \right);
                                  A\imply(B \imply C)
                                \right)}
            }
            {\Seq{B, A \imply (B \imply C)}{A \imply C}}
    }
    {\Seq{A \imply (B \imply C)}{B \imply (A \imply C)}}
  }
  {\Seq{}{(A \imply (B \imply C)) \imply (B \imply (A \imply C))}}       
\end{displaymath}

\begin{center}
$\mathlarger{\mathlarger{\triangledown}}$  
\end{center}

\begin{displaymath}
  \prftree[r]{$\imply\mathrm{r}$}
  {  
    \prftree[r]{$\imply\mathrm{r}$}
    {
      \prftree[r]{$\imply\mathrm{r}$}
            {
             \prfinterspace=8em
             \prftree[r]{$\imply\mathrm{l}$}
                {\Seq{\Gamma}{A}}
                {
                  \prftree[r]{$\imply\mathrm{l}$}
                               {F(B; \{\Gamma, B \imply C)\}}
                               {F(C; \{\Gamma, B \imply C)\}}
                               {\imply-lc(c(F(B; \{\Gamma, B \imply C)\});
                                          c(F(C; \{\Gamma, B \imply C)\});
                                          B \imply C)}       
                }                
                {\Seq{\Gamma}{C}}
            }
            {\Seq{B, A \imply (B \imply C)}{A \imply C}}
    }
    {\Seq{A \imply (B \imply C)}{B \imply (A \imply C)}}
  }
  {\Seq{}{(A \imply (B \imply C)) \imply (B \imply (A \imply C))}}       
\end{displaymath}

\begin{center}
$\mathlarger{\mathlarger{\triangledown}}$  
\end{center}

\begin{equation}
  \label{proof:lj}
  \prftree[r]{$\imply\mathrm{r}$}
  {  
    \prftree[r]{$\imply\mathrm{r}$}
    {
      \prftree[r]{$\imply\mathrm{r}$}
            {
             \prfinterspace=8em
             \prftree[r]{$\imply\mathrm{l}$}
                {\Seq{\Gamma}{A}}
                {
                  \prftree[r]{$\imply\mathrm{l}$}
                               {\Seq{\Gamma, B \imply C}{B}}
                               {\Seq{\Gamma, B \imply C, C}{C}}
                               {\Seq{\Gamma, B \imply C}{C}}       
                }                
                {\Seq{\Gamma}{C}}
            }
            {\Seq{B, A \imply (B \imply C)}{A \imply C}}
    }
    {\Seq{A \imply (B \imply C)}{B \imply (A \imply C)}}
  }
  {\Seq{}{(A \imply (B \imply C)) \imply (B \imply (A \imply C))}}       
\end{equation}
}

\begin{theorem}
  \label{theo:upperbound-lj}
  The size of proofs in \ljarrow~considering only implicational tautologies is
  the same of that in Natural Deduction, i.e. for an implicational formula
  $\alpha$, a proof in \ljarrow~has maximum height of $|\alpha| \cdot
  2^{|\alpha| + 1}$.
\end{theorem}

\begin{proof}
  This proof follows directly from the translation function as each step in the
  Natural Deduction proof is translated into precisely one step in the
  \ljarrow~resultant proof.
\end{proof}

%%%%%%%%%%%%%%%%%%%%%%%%%%%%%%%%%%%%%%%%%%%%%%%%%%%%%%%%%%%%%%%%%%%%%%%%%%%%
\section{The Sequent Calculus \lmt}
\label{sec:lmt}
%%%%%%%%%%%%%%%%%%%%%%%%%%%%%%%%%%%%%%%%%%%%%%%%%%%%%%%%%%%%%%%%%%%%%%%%%%%%

In this section, we present a sound and complete \emph{sequent calculus} for
\mil. We call this system \lmt. We can prove for each rule that if all premises
are valid, then the conclusion is also valid, and if at least one premise is
invalid, then the conclusion also is. This proof is constructive, i.e., 
for any sequent, we have an effective way to produce either a
proof or a counter-model of it.

We start defining the concept of \emph{sequent} used in the proposed calculus. A
sequent in our system has the following general form:
\begin{equation}
  \label{eq:sequent}
  \Seq{\{\Delta'\}, \Upsilon_1^{p_1}, \Upsilon_2^{p_2}, ..., \Upsilon_n^{p_n}, \Delta}{[p_1, p_2,..., p_n], \varphi}
\end{equation} 
\noindent where $\varphi$ is a formula in $\mathcal{L}$ and $\Delta$,
$\Upsilon_1^{p_1}, \Upsilon_2^{p_2}, ..., \Upsilon_n^{p_n}$ are bags\footnote{A
  bag (or a multiset) is a generalization of the concept of a set that, unlike a
  set, takes repetitions into account: a bag \{A, A, B\} is not the same as the
  bag \{A, B\}.} of formulas. Each $\Upsilon_i^{p_i}$ represents formulas
associated with an atomic formula $p_i$.

A sequent has two \emph{focus areas}, one in the left side (curly
bracket)\footnote{Note that the symbols $\{$ and $\}$ here do not represent a
  set. Instead, these symbols work as an annotation in the sequent to determine the left side
  focused area. Therefore, $\Delta'$ instead is a set of formulas in the focused
  area.} and another on the right (square bracket). Curly brackets are used to
control the application of the $\imply$-left rule and square brackets are used
to keep control of formulas that are related to a particular counter-model
definition. $\Delta'$ is a set of formulas and ${p_1, p_2,..., p_n}$ is a
sequence that does not allow repetition. We call \emph{context} of the sequent a
pair $(\alpha, q)$, where $\alpha \in \Delta'$ and $\varphi = q$, where $q$ is
an atomic formula on the right side of the sequent.

The axioms and rules of \lmt~are presented in Figure~\ref{fig:lmtrules}. In each
rule, $\Delta' \subseteq \Delta$.

Rules are inspired by their backward application. In a \imply-left rule
application, the atomic formula, $q$, on the right side of the conclusion goes
to the []-area in the left premise. $\Delta$ formulas in the conclusion are
copied to the left premise and marked with a label relating each of them with
$q$. The left premise also has a copy of $\Delta$ formulas without the
$q$-label. This mechanism keeps track of proving attempts. The form of the
restart rule is better understood in the completeness proof on
Section~\ref{sec:completeness}. A forward reading of rules can be achieved by
considering the notion of validity, as described in Section~\ref{sec:soundness}.

%%%%%%%%%%%%%%%%%%%%%%%%%%%%%%%%%%%%%%
\begin{figure}[h]
\begin{dedsystem}
\hline
%--------------------------
\textbf{Axiom:} \\
%--------------------------
  \AxiomC{}
  \RightLabel{\scriptsize{$axiom$}}
  \UnaryInfC{\Seq{\{\Delta', q\}, \Upsilon_1^{p_1}, \Upsilon_2^{p_2}, \ldots, \Upsilon_n^{p_n}, \Delta}{[p_1, p_2, \ldots, p_n], q}}
  \DisplayProof \\
%--------------------------
\textbf{Focus:} \\
%--------------------------  
  \AxiomC{\Seq{\{\Delta', \alpha\}, \Upsilon_1^{p_1}, \Upsilon_2^{p_2}, \ldots, \Upsilon_n^{p_n}, \Delta, \alpha}{[p_1, p_2, \ldots, p_n], \beta}}
  \RightLabel{\scriptsize{f$_\alpha$}}
  \UnaryInfC{\Seq{\{\Delta'\}, \Upsilon_1^{p_1}, \Upsilon_2^{p_2}, \ldots, \Upsilon_n^{p_n}, \Delta, \alpha}{[p_1, p_2, \ldots, p_n], \beta}}
  \DisplayProof \\
%--------------------------
\textbf{Restart:} \\
%--------------------------  
  \AxiomC{\Seq{\{\}, \Upsilon_1, \Upsilon_2, \ldots, \Upsilon_i, \Upsilon_{i+1}^{p_{i+1}}, \ldots, \Upsilon_n^{p_n}, \Delta^q}{[p_1, p_2, \ldots, p_{i+1}, \ldots, p_n, q], p_i}} 
 \RightLabel{\scriptsize{r$_{p_i}$}}
  \UnaryInfC{\Seq{\{\Delta'\}, \Upsilon_1^{p_1}, \Upsilon_2^{p_2}, \ldots, \Upsilon_i^{p_i}, \Upsilon_{i+1}^{p_{i+1}}, \ldots, \Upsilon_n^{p_n}, \Delta}{[p_1, p_2, \ldots, p_i, p_{i+1}, \ldots, p_n], q}}
  \DisplayProof \\
%--------------------------
\imply\textbf{-Right} \\
%--------------------------  
  \AxiomC{\Seq{\{\Delta'\}, \Upsilon_1^{p_1}, \Upsilon_2^{p_2}, \ldots, \Upsilon_n^{p_n}, \Delta, \alpha}{[p_1, p_2, \ldots, p_n], \beta}}
  \RightLabel{\scriptsize{$\imply$-r$_{\alpha \imply \beta}$}}
  \UnaryInfC{\Seq{\{\Delta'\}, \Upsilon_1^{p_1}, \Upsilon_2^{p_2}, \ldots, \Upsilon_n^{p_n}, \Delta}{[p_1, p_2, \ldots, p_n], \alpha \imply \beta}}
  \DisplayProof \\
%--------------------------
\imply\textbf{-Left} \\
%--------------------------    
Considering $\overline{\Upsilon} = \displaystyle \bigcup_{i=1}^{n}{\Upsilon_i^{p_i}}$ and $\bar{p} = p_1, p_2, \ldots, p_n$, we have:
\\
\AxiomC{\Seq{\{\alpha\imply\beta,\Delta'\},\overline{\Upsilon}, \Delta^q, \Delta}{[\bar{p}, q], \alpha}}
\AxiomC{\Seq{\{\alpha\imply\beta,\Delta'\},\overline{\Upsilon}, \Delta, \beta}{[\bar{p}], q}}
\RightLabel{\scriptsize{$\imply$-l$_{(\alpha \imply \beta, q)}$}}
\BinaryInfC{\Seq{\{\alpha\imply\beta,\Delta'\}\overline{\Upsilon}, \Delta} {[\bar{p}], q}}
\DisplayProof \\ \\ [3ex]\hline
\end{dedsystem}
\caption{Rules of \lmt}
\label{fig:lmtrules}
\end{figure}

%%%%%%%%%%%%%%%%%%%%%%%%%%%%%%%%%%%%%%%%%%%%%%%%%%%%%%%%%%%%%%%%%%%%%%%%%%%%
\subsection{A Proof Search Strategy}
\label{sec:strategy}

The following is a general strategy to be applied with the rules of \lmt~to
generate proofs from an input sequent (a sequent that is a candidate to be the
conclusion of a proof), which is based on a bottom-up application of the rules.
From the proposed strategy, we can then state a proposition about the
termination of the proving process.

A \emph{goal sequent} is a new sequent in the form of~\eqref{eq:sequent}. It is
a premise of one of the system's rules, generated by the application of this
rule on an open branch during the proving process. If the goal sequent is an
axiom, the branch where it is will stop. Otherwise, apply the first applicable
rule in the following order:

\begin{enumerate}
\item Apply \imply-right rule if it is possible, i.e., if the formula on the
  right side of the sequent, outside the []-area, is not atomic. The premise
  generated by this application is the new goal of this branch.
\item Choose one formula on the left side of the sequent, not labeled yet, i.e.,
  a formula $\alpha \in \Delta$ that is not occurring in $\Delta'$, then apply
  the focus rule. The premise generated by this application is the new goal of
  this branch.
\item If all formulas on the left side have already been focused, choose the
  first formula $\alpha \in \Delta'$ such that the context $(\alpha, q)$ was not
  yet tried since the last application of a restart rule. We say that a context
  $(\alpha, q)$ is already tried when a formula $\alpha$ on the left was
  expanded (by the application of \imply-left rule), with $q$ as the formula
  outside the []-area on the right side of the sequent. The premises generated
  by this application are new goals of the respective new branches.
\item Choose the leftmost formula inside the []-area that was not chosen before
  in this branch and apply the restart rule. The premise generated by this
  application is the new goal of the branch.
\end{enumerate}

\begin{observation}
  \label{obs:strategy}
  From the proof strategy we can make the following observations about a tree
  generated during a proving process:
  \begin{enumerate}[(i)]
  \item\label{obs:strategy-1} A \emph{top sequent} is the highest sequent of a
    branch in the tree.
  \item\label{obs:strategy-2} In a top sequent of a branch on the form of
    sequent \eqref{eq:sequent}, if $\varphi \in \Delta$ then the top sequent is
    an axiom and the branch is called a \emph{closed} branch. Otherwise, we say
    that the branch is \emph{open} and $\varphi$ is an atomic formula.
  \item\label{obs:strategy-3} In every sequent of the tree, $\Delta' \subseteq
    \Delta$.
  \item\label{obs:strategy-4} For $i=1, \ldots n$, $\Upsilon_{i-1}^{p_{i-1}}
    \subseteq \Upsilon_{i}^{p_{i}}$.
  \end{enumerate}
\end{observation}

We call this strategy $\mathcal{S}$-\lmt. Figure~\ref{fig:generalprooftree}
shows a generic schema of a completely expanded branch in a tree, not
necessarily a proof, generated by the application of $\mathcal{S}$-\lmt.

%%%%%%%%%%%%%%%%%%%%%%%%%%%%%%%%%%%%%%
%\begin{landscape}

\hspace{-5em}

\begin{figure}[h]
  \advance\leftskip-2cm
%\begin{dedsystem}
%  \hline
%\footnotesize
{
  \AxiomC{\Seq{\{\Delta'\}, \Upsilon_1^{p_1}, \Upsilon_2^{p_2}, \Upsilon_{i-1}^{p_{i-1}},\ldots,\Upsilon_i^{p_i},\Delta}{[p_1, p_2, \ldots, p_{i-1}, p_i], p_k \quad (\emph{where k = 1, 2, \ldots, i})}}
  \def\defaultHypSeparation{\hskip -.2in}
  \noLine
  \UnaryInfC{\vdots}
  \noLine
  \UnaryInfC{\emph{a sequence of focus, \imply-left and \imply-right}}
  \UnaryInfC{\Seq{\{\}, \Upsilon_1^{p_1}, \Upsilon_2^{p_2}, \ldots, \Upsilon_{i-1}^{p_{i-1}},\Delta}{[p_1, p_2, \ldots, p_{i-1}], p_i}}
  \noLine
  \UnaryInfC{\vdots}
  \noLine
  \UnaryInfC{\emph{a sequence of focus, \imply-left, \imply-right and restart (for each atomic formula in the []-area)}}
  \UnaryInfC{\Seq{\{\varphi\imply\psi\},\Upsilon_2^{p_2},\ldots,\Upsilon_{i-1}^{p_{i-1}},\Delta^{p_i}, \Upsilon_1^{p_1}, \Upsilon_1, \varphi_1,\ldots,\varphi_n}{[p_2, \ldots, p_{i-1}, p_i, p_1], p_2}}
  \UnaryInfC{\vdots}
  \RightLabel{$\imply$-right}
  \UnaryInfC{\Seq{\{\varphi\imply\psi\},\Upsilon_2^{p_2},\ldots,\Upsilon_{i-1}^{p_{i-1}},\Delta^{p_i}, \Upsilon_1^{p_1}, \Upsilon_1}{[p_2, \ldots, p_{i-1}, p_i, p_1], \varphi}} 
                                                                              \AxiomC{\vdots}
  \RightLabel{$\imply$-left}
  \BinaryInfC{\Seq{\{\varphi\imply\psi\}, \Upsilon_2^{p_2},\ldots,\Upsilon_{i-1}^{p_{i-1}},\Delta^{p_i}, \Upsilon_1}{[p_2, \ldots, p_{i-1}, p_i], p_1}}
  \RightLabel{focus}
  \UnaryInfC{\Seq{\{\},\Upsilon_2^{p_2},\ldots,\Upsilon_{i-1}^{p_{i-1}},\Delta^{p_i}, \Upsilon_1}{[p_2, \ldots, p_{i-1}, p_i], p_1}}
  \RightLabel{restart-$p_1$}
  \UnaryInfC{\Seq{\{\Delta'\}, \Upsilon_1^{p_1}, \Upsilon_2^{p_2},\ldots,\Upsilon_{i-1}^{p_{i-1}},\Delta}{[p_1, p_2, \ldots, p_{i-1}], p_i}}
  \noLine
  \UnaryInfC{\vdots}
  \noLine
  \UnaryInfC{\emph{a sequence of focus, \imply-left and \imply-right}}
  \UnaryInfC{\Seq{\{\varphi\imply\psi\}, \Upsilon_1^{p_1}, \Upsilon_1, \varphi_1,\ldots,\varphi_n}{[p_1], p_2}}
  \noLine
  \UnaryInfC{\vdots}
  \noLine
  \UnaryInfC{\emph{a sequence \imply-right}}
  \UnaryInfC{\Seq{\{\varphi\imply\psi\}, \Upsilon_1^{p_1}, \Upsilon_1}{[p_1], \varphi}} 
                                                                              \AxiomC{\vdots}
                                                                     \UnaryInfC{\Seq{\{\varphi\imply\psi, \psi\}, \psi, \Upsilon_1}{[], p_1}}
                                                                     \RightLabel{focus}
                                                                     \UnaryInfC{\Seq{\{\varphi\imply\psi \}, \psi, \Upsilon_1}{[], p_1}}
  \RightLabel{$\imply$-left}
  \BinaryInfC{\Seq{\{\varphi \imply \psi\}, \Upsilon_1}{[], p_1}}
  \RightLabel{focus}
  \UnaryInfC{\Seq{\{\}, \varphi\imply\psi, \gamma_1, \ldots, \gamma_m}{[], p_1}}
  \noLine
  \UnaryInfC{\vdots}
  \noLine
  \UnaryInfC{\Seq{\{\}, \varphi\imply\psi}{[], \gamma }}
  \RightLabel{$\imply$-right}
  \UnaryInfC{\Seq{\{\}}{[], (\varphi\imply\psi)\imply\gamma }}
\DisplayProof %\\ [3ex]\hline
%\end{dedsystem}
}
\caption{Generic schema of a tree generated following $\mathcal{S}$-\lmt}
\label{fig:generalprooftree}
\end{figure}
% \end{landscape}

\vspace{-0.5cm}

%%%%%%%%%%%%%%%%%%%%%%%%%%%%%%%%%%%%%%%%%%%%%%%%%%%%%%%%%%%%%%%%%%%%%%%%%%%%
\subsection{An Upper Bound for the Proof Search in \lmt}

Using the same approach applied in Section~\ref{sec:proofsize}, we now propose a
translation from \ljarrow~proofs into the system~\lmt. The translation function
needs to adapt a sequent in \ljarrow~form to a sequent in \lmt~form.
Figure~\ref{fig:lj2lmt} presents the definition of the translation
function\footnote{As in Section~\ref{sec:proofsize}, we use semicolon to separate
  function arguments here}.

%%%%%%%%%%%%%%%%%%%%%%%%%%%%%%%%%%%%%%
\begin{figure}[h]
\begin{dedsystem}
\hline
{\small
%--------------------------
\textbf{Axioms:}
%--------------------------
\begin{equation*}
    F'\left(\Seq{\Gamma, \alpha}{\alpha}; \Delta; \Upsilon; \Sigma; \Pi\right) = 
    \alwaysNoLine
    \AxiomC{$\Seq{\{\Delta\}, \Upsilon, \Gamma, \alpha}{[\Sigma], \alpha}$}
    \def\extraVskip{2pt}
    \UnaryInfC{$\prod$}
    \DisplayProof
\end{equation*}

\bigskip

%--------------------------
\textbf{Last rule is \imply-right:}
%--------------------------
\begin{equation*}
F'\left(
  \alwaysNoLine
  \AxiomC{$\mathcal{D}$}
  \UnaryInfC{$\Seq{\Gamma, \alpha}{\beta}$}
  \alwaysSingleLine
  \RightLabel{$\imply$-r}
  \UnaryInfC{$\Seq{\Gamma}{\alpha \imply \beta}$}
  \DisplayProof;
  \Delta;  
  \Upsilon;
  \Sigma;
  \Pi
\right) = 
\end{equation*}

\begin{minipage}{1\textwidth}\hspace{3.5cm}
\prftree[noline]
{
\prftree[r]{$\imply\mathrm{r}$}
             {F'\left(
                  \alwaysNoLine
                  \AxiomC{$\mathcal{D}$}
                  \UnaryInfC{$\Seq{\Gamma, \alpha}{\beta}$}
                  \DisplayProof;
                  \Delta;
                  \Upsilon;
                  \Sigma;
                  \alwaysNoLine
                  \AxiomC{$\Seq{\{\Delta\}, \Upsilon, \Gamma}{[\Sigma], \alpha\imply\beta}$}
                  \UnaryInfC{$\prod$}
                  \DisplayProof
                \right)}
             {\Seq{\{\Delta\}, \Upsilon, \Gamma}{[\Sigma], \alpha\imply\beta}}
}
{\mathlarger{\mathlarger{\Pi}}}
\end{minipage}

\bigskip

%--------------------------
\textbf{Last rule is \imply-left:}
%--------------------------
\begin{equation*}
F'\left(
  \alwaysNoLine
  \AxiomC{$\mathcal{D}_1$}
  \UnaryInfC{$\Seq{\Gamma, \alpha \imply \beta}{\alpha}$}
            \AxiomC{$\mathcal{D}_2$}
            \UnaryInfC{$\Seq{\Gamma, \beta}{q}$}
  \alwaysSingleLine
  \RightLabel{$\imply$-l}            
  \BinaryInfC{$\Seq{\Gamma, \alpha \imply \beta}{q}$}
  \DisplayProof;
  \Delta;
  \Upsilon;
  \Sigma;
  \Pi
 \right) =
\end{equation*}

\begin{minipage}{1\textwidth}
  \prfinterspace=8em             
  \prftree[r]{$\imply\mathrm{l}$}
  {\mathcal{D}_1'}
  {\mathcal{D}_2'}  
  {PROOFUNTIL\left(
      FOCUS\left(
          \alwaysNoLine
          \AxiomC{$\Seq{\{\Delta\}, \Upsilon, \Gamma, \alpha \imply \beta}{[\Sigma], q}$}
          \UnaryInfC{$\mathlarger{\mathlarger{\Pi}}$}
          \DisplayProof
        \right)
    \right)}
\end{minipage}
}
\\ \\ [3ex]\hline
\end{dedsystem}
\caption{A recursive function to translate \ljarrow~into \lmt}
\label{fig:lj2lmt}
\end{figure}
%%%%%%%%%%%%%%%%%%%%%%%%%%%%%%%%%%%%%%

We use some abbreviations to shorten the function definition of
Figure~\ref{fig:lj2lmt}. We present them below.

{\small
\begin{eqnarray*}
\hspace{-6cm}
\begin{minipage}{0.5\textwidth}
    $\mathcal{D}_1'$    
\end{minipage} & \hspace{-6.8cm}= &
\begin{minipage}{0.5\textwidth}\vspace{-.1cm}
                         $F'\left(
                         \alwaysNoLine
                         \AxiomC{$\mathcal{D}_1$}
                         \UnaryInfC{$\Seq{\Gamma, \alpha \imply \beta}{\alpha}$}
                         \DisplayProof;
                         \Delta \cup \{\alpha \imply \beta\};
                         \Upsilon \cup \Gamma^q \cup \{(\alpha\imply\beta)^q\};
                         \Sigma \cup \{q\};
                         \Pi'
                         \right)$
\end{minipage}  
\end{eqnarray*}

\begin{eqnarray*}
\hspace{-6cm}
\begin{minipage}{0.5\textwidth}
    $\mathcal{D}_2'$    
\end{minipage} & \hspace{-6.8cm}= &
\begin{minipage}{0.5\textwidth}\vspace{-.1cm}
                         $F'\left(
                         \alwaysNoLine
                         \AxiomC{$\mathcal{D}_2$}
                         \UnaryInfC{$\Seq{\Gamma, \beta}{q}$}
                         \DisplayProof;
                         \Delta \cup \{\alpha \imply \beta\};
                         \Upsilon;
                         \Sigma;
                         \Pi'
                         \right)$
\end{minipage}  
\end{eqnarray*}

\begin{eqnarray*}
\hspace{-6cm}
\begin{minipage}{0.5\textwidth}
    $\Pi'$    
\end{minipage} & \hspace{-6.8cm}= &
\begin{minipage}{0.5\textwidth}\vspace{-.1cm}
$PROOFUNTIL\left(
      FOCUS\left(
          \alwaysNoLine
          \AxiomC{$\Seq{\{\Delta\}, \Upsilon, \Gamma, \alpha \imply \beta}{[\Sigma], q}$}
          \UnaryInfC{$\mathlarger{\mathlarger{\Pi}}$}
          \DisplayProof
        \right)
    \right)$
\end{minipage}  
\end{eqnarray*}

\begin{eqnarray*}
\hspace{-3cm}
\begin{minipage}{0.3\textwidth}
    $\Gamma^{q}$   
\end{minipage} & \hspace{-4cm}= &
\begin{minipage}{0.7\textwidth}\vspace{-.1cm}
means that all formulas of the set $\Gamma$ are labeled with a reference to the atomic formula $q$
\end{minipage}  
\end{eqnarray*}

\begin{eqnarray*}
\hspace{-2cm}
\begin{minipage}{0.3\textwidth}
    $(\alpha\imply\beta)^q$    
\end{minipage} & \hspace{-3cm}= &
\begin{minipage}{0.7\textwidth}\vspace{-.1cm}
means the same for the individual formula $\alpha \imply \beta$.
\end{minipage}  
\end{eqnarray*}
}
%%%%%%%%%%%%%%%%%%%%%%%%%%%%%%%%%%%%%%

The complicated case occurs when the function $F'$ is applied to a proof fragment in
which the last \ljarrow~rule applied is an \imply-left. In this case, $F'$ needs
to inspect the proof fragment constructed until that point to identify whether
the context $(\alpha\imply\beta, q)$ was already used or not. This inspection
has to be done since \lmt~does not allow two or more applications of the same
context between two applications of the $restart$ rule. To deal with this, we
use some auxiliary functions described below.

$FOCUS$ is a function that receives a fragment of proof in \lmt~form and builds
one application of the focus rule on the top of the proof fragment received in
the case that the main formula of the rule is not already focused. The main
formula is also an argument of the function. In the function definition
\eqref{eq:focus-function}, we have the constraint that $\alpha \in \Gamma$.

\newpage

\hspace{-5cm}
\begin{minipage}{1\textwidth}
  \begin{equation*}   
  FOCUS\left(
    \alwaysNoLine
    \AxiomC{$\Seq{\{\Delta\}, \Upsilon, \Gamma}{[\Sigma], \beta}$}
    \def\extraVskip{2pt}
    \UnaryInfC{$\prod$}
    \UnaryInfC{$\Seq{\{\}}{[], \gamma}$}
    \DisplayProof;
    \alpha
  \right) = 
  \end{equation*}
\end{minipage}

\bigskip\bigskip\bigskip

\begin{minipage}{1\textwidth}  
  \begin{equation}
    \label{eq:focus-function}
    \left\{
      \begin{array}{lr}
        \alwaysSingleLine
        \AxiomC{$\Seq{\{\Delta, \alpha\}, \Upsilon, \Gamma}{[\Sigma], \beta}$}
        \RightLabel{$\mathrm{focus}$}
        \UnaryInfC{$\Seq{\{\Delta\}, \Upsilon, \Gamma}{[\Sigma], \beta}$}
        \alwaysNoLine        
        \def\extraVskip{2pt}
        \UnaryInfC{$\prod$}
        \UnaryInfC{$\Seq{\{\}}{[], \gamma}$}
        \DisplayProof & \text{if }\alpha \notin \Delta \\
        \\
        \alwaysNoLine
        \AxiomC{$\Seq{\{\Delta\}, \Upsilon, \Gamma}{[\Sigma], \beta}$}
        \def\extraVskip{2pt}
        \UnaryInfC{$\prod$}
        \UnaryInfC{$\Seq{\{\}}{[], \gamma}$}
        \DisplayProof & \text{otherwise} \\        
        \end{array}
    \right\}
  \end{equation}
\end{minipage}

%%%%%%%%%%%%%%%%%%%%%%%%%%%%%%%%%%%%%%
\bigskip\bigskip

%%%%%%%%%%%%%%%%%%%%%%%%%%%%%%%%%%%%%%
The function $PROOFUNTIL$ also receives a fragment of a proof in \lmt~form where
$(\alpha \imply \beta, q)$ is one of the available contexts, applies the restart
rule with an atomic formula $p$ such that $p \in \Sigma$ in the top of this
fragment of proof and, then, conducts a sequence of \lmt~rule applications
following the $\mathcal{S}$-\lmt  until the point that the context $(\alpha
\imply \beta, q)$ is available again. This mechanism has to be done in the case
that the context $(\alpha \imply \beta, q)$ is already applied in an
$\imply$-left application, some point after the last restart rule application in
the proof fragment received as the argument $\Pi$. Otherwise, the proof fragment
is returned unaltered. Function $PROOFUNTIL$ is described in the function
definition~\eqref{eq:proofuntil-function}.

\bigskip

\hspace{-3.5cm}
\begin{minipage}{1\textwidth}
  \begin{equation*}
  PROOFUNTIL\left(
    \alwaysNoLine
    \AxiomC{$\Seq{\{\Delta, \alpha\imply\beta\}, \Upsilon, \Gamma}{[\Sigma], q}$}
    \def\extraVskip{2pt}
    \UnaryInfC{$\prod$}
    \UnaryInfC{$\Seq{\{\}}{[], \gamma}$}
    \DisplayProof
  \right) = 
  \end{equation*}
\end{minipage}

\bigskip\bigskip\bigskip

\begin{minipage}{1\textwidth}
  \begin{equation}
    \label{eq:proofuntil-function}    
    \left\{
      \begin{array}{lr}
        \alwaysSingleLine
        \AxiomC{$\Seq{\{\Delta'', \alpha\imply\beta\}, \Upsilon'', \Gamma'}{[\Sigma''], q}$}
        \UnaryInfC{\vdots}
        \UnaryInfC{$\Seq{\{\}, \Upsilon', \Gamma^q, \Gamma}{[\Sigma'], p}$}        
        \RightLabel{$\mathrm{restart-}p$}        
        \UnaryInfC{$\Seq{\{\Delta, \alpha\imply\beta\}, \Upsilon, \Gamma}{[\Sigma], q}$}
        \alwaysNoLine          
        \def\extraVskip{2pt}
        \UnaryInfC{$\prod$}
        \UnaryInfC{$\Seq{\{\}}{[], \gamma}$}
        \DisplayProof & \imply\text{-left}(\alpha\imply\beta, q) \in \prod\\        
        \\
        \bigskip
        \\
        \alwaysNoLine
        \AxiomC{$\Seq{\{\Delta, \alpha\imply\beta\}, \Upsilon, \Gamma}{[\Sigma], q}$}
        \def\extraVskip{2pt}
        \UnaryInfC{$\prod$}
        \UnaryInfC{$\Seq{\{\}}{[], \gamma}$}
        \DisplayProof & \text{otherwise} \\        
     \end{array}
    \right\}      
  \end{equation}
\end{minipage}

%%%%%%%%%%%%%%%%%%%%%%%%%%%%%%%%%%%%%%

\bigskip\bigskip\bigskip

\bigskip\

As an example of this translation, we use here the formula ${(((( A \to B ) \to
  A ) \to A ) \to B ) \to B}$. We know from \cite{Dowek2006} that this formula
needs two repetitions of the hypothesis ${((( A \to B ) \to A ) \to A ) \to B}$
to be proved in \mil. To shorten the proof tree we use the following
abbreviation: ${(( A \to B ) \to A ) \to A = \epsilon}$. Thus, its normal proof
in Natural Deduction can be represented as shown in
Proof~\eqref{proof:nd-twice}.

\bigskip

%%%%%%%%%%%%%%%%%%%%%%%%%%%%%%%%%%%%%%
%LJ2LMT: Natural Deduction Proof
%%%%%%%%%%%%%%%%%%%%%%%%%%%%%%%%%%%%%%
{\small
\begin{equation}
\label{proof:nd-twice}    
\infer[\to \mbox{I}^3]
{(\epsilon \to  B ) \to  B}
{\infer[\to \mbox{E}]
{ B }
{\infer[\to \mbox{E}^2]
{\epsilon}
{\infer[\to \mbox{E}]
{ A }
{\infer[\to \mbox{I}^1]
{( A  \to  B )}
{\infer[\to \mbox{E}]
{ B }
{\infer[\to \mbox{I}^4]
{\epsilon}
{{[ A ]^1}
}
&
{[\epsilon \to  B]^3}
}
}
&
{[( A  \to  B ) \to  A]^2}
}
}
&
{[\epsilon \to  B]^3}
}
}
\end{equation}
}

Using the translation presented in Figure~\ref{fig:nd2lj}, we achieve the
following \ljarrow~cut-free proof. To shorten the proof, we numbered each
subformula of the initial formula that we want to prove and use these numbers to
refer to those subformulas all over the proof. Let us call this \ljarrow~version
of the proof $\mathcal{D}$ (Proof~\eqref{proof:lj-twice}).

%%%%%%%%%%%%%%%%%%%%%%%%%%%%%%%%%%%%%%
%LJ2LMT: Natural Deduction 2 LJ
%%%%%%%%%%%%%%%%%%%%%%%%%%%%%%%%%%%%%%

{\small
\begin{equation}
\label{proof:lj-twice}    
  \alwaysSingleLine
  \def\extraVskip{2pt}
  \AxiomC{$\Seq{\to_3, \to_1, A, \to_1}{A}$}
  \RightLabel{$\imply$-R}  
  \UnaryInfC{$\Seq{\to_3, \to_1, A}{\to_2}$}
                                               \AxiomC{$\Seq{\to_3, \to_1, A, B}{B}$}
  \RightLabel{$\imply$-L}   
  \BinaryInfC{$\Seq{\to_3, \to_1, A}{B}$}
  \RightLabel{$\imply$-R}
  \UnaryInfC{$\Seq{\to_3, \to_1}{\to_0}$}  
                                               \AxiomC{$\Seq{\to_3, \to_1, A}{A}$}
  \RightLabel{$\imply$-L}  
  \BinaryInfC{$\Seq{\to_3, \to_1}{A}$}
  \RightLabel{$\imply$-R}
  \UnaryInfC{$\Seq{\to_3}{\to_2}$}
                                               \AxiomC{$\Seq{\to_3, B}{B}$}
  \RightLabel{$\imply$-L}  
  \BinaryInfC{$\Seq{\to_3}{B}$}
  \RightLabel{$\imply$-R}
  \UnaryInfC{$\Seq{}{(((( A  \to_0  B ) \to_1  A ) \to_2  A ) \to_3  B ) \to_4  B}$}
  \DisplayProof
\end{equation}
}

\bigskip

The translation to \lmt~starts by applying the function $F'$ to the full proof $\mathcal{D}$.

%%%%%%%%%%%%%%%%%%%%%%%%%%%%%%%%%%%%%%
%LJ2LMT: Step 1
%%%%%%%%%%%%%%%%%%%%%%%%%%%%%%%%%%%%%%
{\small
\begin{equation*}
  F'\left(
    \begin{array}{llcr}
      \multirow{5}{*}{$\mathcal{D}$;} & \Gamma   &= & \emptyset;\\
                                      & \Delta   &= & \emptyset;\\
                                      & \Upsilon &= & \emptyset;\\
                                      & \Sigma   &= & \emptyset;\\
                                      & \Pi      &= & nil
    \end{array}
    \right)
\end{equation*}
}

\bigskip

This first call of $F'$ produces the end sequent of the proof in \lmt~and calls
the function $F'$ recursively to the rest of the original proof in \ljarrow.
This \ljarrow~fragment has now, as its last rule application, an \imply-left.

%%%%%%%%%%%%%%%%%%%%%%%%%%%%%%%%%%%%%%
%LJ2LMT: Step 2
%%%%%%%%%%%%%%%%%%%%%%%%%%%%%%%%%%%%%%
{\small

\begin{displaymath}
  \begin{array}{lcl}
    \prod & = & \Seq{\{\}}{[], (((( A  \to_0  B ) \to_1  A ) \to_2  A ) \to_3  B ) \to_4  B}
  \end{array}    
\end{displaymath}

\begin{displaymath}
  \prftree[r]{$\imply\mathrm{r}$}
                {F'\left(
                    \begin{array}{llcc}
                      \multirow{5}{*}{
                      \alwaysNoLine
                      \def\extraVskip{2pt}
                      \AxiomC{$\mathcal{D}_1$}
                      \UnaryInfC{$\Seq{\to_3}{\to_2}$}
                                                 \AxiomC{$\Seq{\to_3, B}{B}$}
                      \alwaysSingleLine                                                 
                      \RightLabel{$\imply$-l}
                      \BinaryInfC{$\Seq{\to_3}{B}$}
                      \DisplayProof; } & \Gamma   &= &\emptyset;\\
                                       & \Delta   &= &\emptyset;\\
                                       & \Upsilon &= &\emptyset;\\
                                       & \Sigma   &= &\emptyset;\\
                                       & \Pi & & 
                    \end{array}
                   \right)
                }                
                {\prod}
\end{displaymath}          
}

Then, as the main formula $((( A \to_0 B ) \to_1 A ) \to_2 A ) \to_3 B$ (in the
proof represented only by $\to_3$) is not focused yet, the call of function $F'$
first constructs an application of the focus rule on the top of the $\Pi$
fragment received as a function argument. Also, the context $(\to_3, B)$ was not
expanded yet. Thus, the recursive step can proceed directly without the need of
a restart (this is controlled by the $PROOFUNTIL$ function as shown in $F'$
definition in Figure~\ref{fig:lj2lmt}).

%%%%%%%%%%%%%%%%%%%%%%%%%%%%%%%%%%%%%%
%LJ2LMT: Step 3
%%%%%%%%%%%%%%%%%%%%%%%%%%%%%%%%%%%%%%
{\small
\begin{displaymath}
  \begin{array}{lcl}
    \prod & = &
              \prftree[r]{$\imply\mathrm{r}$}
              {
                  \prftree[r]{$\imply\mathrm{focus}$}
                  {\Seq{\{\to_3\},\to_3}{[], B}}
                  {\Seq{\{\},\to_3}{[], B}}
              }
              {\Seq{\{\}}{[], (((( A  \to_0  B ) \to_1  A ) \to_2  A ) \to_3  B ) \to_4 B}}              
  \end{array}    
\end{displaymath}

\bigskip\bigskip

\begin{displaymath}
  \prftree[r]{$\imply\mathrm{l}$}
  {F'\left(
      \begin{array}{llcc}
        \multirow{5}{*}{
        \alwaysNoLine
        \def\extraVskip{2pt}
        \AxiomC{$\mathcal{D}_2$}
        \UnaryInfC{$\Seq{\to_3}{\to_2}$}
        \DisplayProof; } & \Gamma   &= &\{\to_3\};\\
                         & \Delta   &= &\{\to_3\};\\
                         & \Upsilon &= &\{\to_3^B\};\\
                         & \Sigma   &= &\{B\};\\
                         & \Pi & & 
      \end{array}
    \right)
  }
  {F'\left(
      \begin{array}{llcc}
        \multirow{5}{*}{
        \alwaysNoLine
        \def\extraVskip{2pt}
        \AxiomC{$\Seq{\to_3, B}{B}$}
        \DisplayProof; } & \Gamma   &= &\{\to_3\};\\
                         & \Delta   &= &\{\to_3\};\\
                         & \Upsilon &= &\emptyset;\\
                         & \Sigma   &= &\emptyset;\\
                         & \Pi & &
      \end{array}
    \right)
  }
  {\prod}
\end{displaymath}
}

The call of $F'$ on the right premise constructs an axiom. Thus this branch in
the \lmt~proof translation being built is closed. The next recursive call
"pastes" on the top of the right branch of the new version of $\Pi$ as follows.

\bigskip

%%%%%%%%%%%%%%%%%%%%%%%%%%%%%%%%%%%%%%
%LJ2LMT: Step 4
%%%%%%%%%%%%%%%%%%%%%%%%%%%%%%%%%%%%%%
{\small

\begin{displaymath}
  \begin{array}{lcl}
    \prod & = &
              \prftree[r]{$\imply\mathrm{r}$}
              {
                  \prftree[r]{$\imply\mathrm{focus}$}
                  {
                      \prftree[r]{$\imply\mathrm{l}$}
                      {\Seq{\{\to_3\}, \to_3^B, \to_3}{[B], \to_2}}                 
                      {\Seq{\{\to_3\}, \to_3, B}{[], B}}
                      {\Seq{\{\to_3\},\to_3}{[], B}}
                  }
                  {\Seq{\{\},\to_3}{[], B}}
              }
              {\Seq{\{\}}{[], (((( A  \to_0  B ) \to_1  A ) \to_2  A ) \to_3  B ) \to_4 B}}
  \end{array}    
\end{displaymath}

\bigskip
  
\begin{displaymath}
  \prftree[r]{$\imply\mathrm{r}$}
  {F'\left(
      \begin{array}{llcc}
        \multirow{5}{*}{
        \alwaysNoLine
        \def\extraVskip{2pt}
        \AxiomC{$\mathcal{D}_3$}
        \UnaryInfC{$\Seq{\to_3, \to_1}{\to_0}$}
        \AxiomC{$\Seq{\to_3, A}{A}$}
        \alwaysSingleLine                                                 
        \RightLabel{$\imply$-L}
        \BinaryInfC{$\Seq{\to_3, \to_1}{A}$}
        \DisplayProof; } & \Gamma   &= &\{\to_3\};\\
                         & \Delta   &= &\{\to_3\};\\
                         & \Upsilon &= &\{\to_3^B\};\\
                         & \Sigma   &= &\{B\};\\
                         & \Pi & & 
      \end{array}
    \right)                 
  }
  {\prod}                 
\end{displaymath}
}

This process goes until a point where the context ${(\to_3, B)}$ is again found,
and the translation needs to deal with a restart in the \lmt~translated proof.
This situation happens when the recursion of $F'$ reaches the point below.

{\small
\begin{displaymath}
  F'\left(
      \begin{array}{llcc}
        \multirow{5}{*}{
        \alwaysNoLine
        \def\extraVskip{2pt}
        \AxiomC{$\mathcal{D}_4$}
        \UnaryInfC{$\Seq{\to_3, \to_1, A}{\to_2}$}
        \AxiomC{$\Seq{\to_3, \to_1, A, B}{B}$}
        \alwaysSingleLine                                                 
        \RightLabel{$\imply$-L}
        \BinaryInfC{$\Seq{\to_3, \to_1, A}{B}$}
        \DisplayProof; } & \Gamma   &= &\{\to_1, A\};\\
                         & \Delta   &= &\{\to_3, \to_1\};\\
                         & \Upsilon &= &\{\to_3^B, \to_3^A, \to_1^A\};\\
                         & \Sigma   &= &\{B, A\};\\
                         & \Pi & & 
      \end{array}
    \right)                            
\end{displaymath}
}

\bigskip
The $\Pi$ fragment constructed in this step of the recursion is presented below. 

%%%%%%%%%%%%%%%%%%%%%%%%%%%%%%%%%%%%%%
%LJ2LMT: Step 5
%%%%%%%%%%%%%%%%%%%%%%%%%%%%%%%%%%%%%%
{\small
  \begin{displaymath}
  \infer[\imply\mathrm{r}]
  {\Seq{\{\}}{[], (((( A  \to_0  B ) \to_1  A ) \to_2  A ) \to_3  B ) \to_4 B}}
  {
    \infer[\imply\mathrm{focus}]
    {\Seq{\{\},\to_3}{[], B}}
    {
      \infer[\imply\mathrm{l}]
      {\Seq{\{\to_3\},\to_3}{[], B}}              
      {
        \infer[\imply\mathrm{r}]
        {\Seq{\{\to_3\}, \to_3^B, \to_3}{[B], \to_2}}                            
        {
          \infer[\imply\mathrm{focus}]
          {\Seq{\{\to_3\}, \to_3^B, \to_3, \to_1}{[B], A}}
          {
            \infer[\imply\mathrm{l}]
            {\Seq{\{\to_3, \to_1\}, \to_3^B, \to_3, \to_1}{[B], A}}              
            {
              \deduce
              {\Seq{\{\to_3, \to_1\}, \to_3^B, \to_3^A, \to_1^A, \to_3, \to_1}{[B, A], \to_0}}
              {\prod'}              
              &
              {\Seq{\{\to_3, \to_1\}, \to_3^B, \to_3, \to_1, A}{[B], A}}
            }
          }
        }
        &                          
        {\Seq{\{\to_3\}, \to_3, B}{[], B}}}
    }                 
  }              
\end{displaymath}
}

The $\Pi'$ fragment in $\Pi$ is built as the result of a proof search from the
top sequent of the leftmost branch of $\Pi$ until the point that there is a
repetition of the context ${(\to_3, B)}$ and the $\imply$-left rule can be
applied again without offending the \lmt~strategy. This result is produced by
the $PROOFUNTIL$ call when applying $F'$ to an \ljarrow~fragment that end with a
\imply-left rule application.

\bigskip

{\small
\begin{displaymath}
  \alwaysSingleLine
  \def\extraVskip{2pt}
  \AxiomC{$\prod_1'$}                                  
                                         \AxiomC{$\prod_2'$}                                  
  \RightLabel{$\imply$-l}  
  \BinaryInfC{$\Seq{\{\to_3\}, \to_3^B, \to_3, \to_1, \to_1^B, A^B, \to_3^B}{[B], A}$}  
  \RightLabel{focus}  
  \UnaryInfC{$\Seq{\{\}, \to_3^B, \to_3, \to_1, \to_1^B, A^B, \to_3^B}{[B], A}$}
  \RightLabel{restart}
  \UnaryInfC{$\Seq{\{\to_3, \to_1\}, \to_3^B, \to_3^A, \to_1^A, \to_1, A, \to_3}{[B, A], B}$}
  \DisplayProof  
\end{displaymath}
}

\noindent where $\Pi_1'$ is:

{\scriptsize
\begin{displaymath}
  \alwaysSingleLine
  \def\extraVskip{2pt}
  \AxiomC{$\Seq{\{\to_3, \to_1\}, \to_3^B, \to_3^A, \to_1^A, \to_1^B, A^B, \to_3^B, \to_3, \to_1, A}{[B, A], B}$}
  \RightLabel{$\imply$-r}
  \UnaryInfC{$\Seq{\{\to_3, \to_1\}, \to_3^B, \to_3^A, \to_1^A, \to_1^B, A^B, \to_3^B, \to_3, \to_1}{[B, A], \to_0}$}
                                            \AxiomC{$\Seq{\{\to_3\}, \to_3^B, \to_3^A, \to_1^A, \to_1^B, A^B, \to_3^B, \to_3, \to_1, A}{[B, A], A}$}  
  \RightLabel{$\imply$-l}  
  \BinaryInfC{$\Seq{\{\to_3, \to_1\}, \to_3^B, \to_3^A, \to_1^A, \to_1^B, A^B, \to_3^B, \to_3, \to_1}{[B, A], A}$}
  \RightLabel{focus}    
  \UnaryInfC{$\Seq{\{\to_3\}, \to_3^B, \to_3^A, \to_1^A, \to_1^B, A^B, \to_3^B, \to_3, \to_1}{[B, A], A}$}   
  \RightLabel{$\imply$-r}    
  \UnaryInfC{$\Seq{\{\to_3\}, \to_3^B, \to_3^A, \to_1^A, \to_1^B, A^B, \to_3^B, \to_3, \to_1}{[B, A], \to_2}$}                                  
  \DisplayProof  
\end{displaymath}
}

\noindent and $\Pi_2'$ is:

{\scriptsize
\begin{displaymath}
  \alwaysSingleLine
  \def\extraVskip{2pt}              
  \AxiomC{$\Seq{\{\to_3, \to_1\}, \to_3^B, \to_3^A, \to_1^A, \to_1^B, A^B, \to_3^B, B^A, \to_3, \to_1, B, A}{[B, A], B}$}
  \RightLabel{$\imply$-r}
  \UnaryInfC{$\Seq{\{\to_3, \to_1\}, \to_3^B, \to_3^A, \to_1^A, \to_1^B, A^B, \to_3^B, B^A, \to_3, \to_1, B}{[B, A], \to_0}$}
  \AxiomC{$\Seq{\{\to_3, to_1\}, \to_3^B, \to_3, \to_1, \to_1^B, A^B, \to_3^B, B, A}{[B], A}$}
  \RightLabel{$\imply$-l}
  \BinaryInfC{$\Seq{\{\to_3, \to_1\}, \to_3^B, \to_3, \to_1, \to_1^B, A^B, \to_3^B, B}{[B], A}$}
  \RightLabel{focus}  
  \UnaryInfC{$\Seq{\{\to_3\}, \to_3^B, \to_3, \to_1, \to_1^B, A^B, \to_3^B, B}{[B], A}$}                                  
  \DisplayProof  
\end{displaymath}
}

The top sequent of the fragment $\Pi_1'$ is the point where we can apply the
$\imply$-left rule to the context ${(\to_3, B)}$ again. Thus the next recursion
call becomes:

\begin{displaymath}
  \prftree[r]{$\imply\mathrm{l}$}
  {\Pi_3'}
  {\Pi_4'}
  {\Pi}                 
\end{displaymath}

\noindent where

{\small
\begin{eqnarray*}
\hspace{-6cm}
\begin{minipage}{0.5\textwidth}
    $\prod_3'$    
\end{minipage} & \hspace{-6.8cm}= &
\begin{minipage}{0.5\textwidth}\vspace{-.1cm}
  $F'\left(
    \begin{array}{llcc}
      \multirow{5}{*}{
      \alwaysSingleLine
      \def\extraVskip{2pt}
      \AxiomC{$\Seq{\to_3, \to_1, A}{A}$}
      \UnaryInfC{$\Seq{\to_3, \to_1, A}{\to_2}$}
      \DisplayProof; } & \Gamma   &= &\{\to_3, \to_1, A\};\\
                       & \Delta   &= &\{\to_3, \to_1\};\\
                       & \Upsilon &= &\{\to_3^A, \to_1^A, \to_3^B, \to_1^B, A^B\};\\
                       & \Sigma   &= &\{B, A\};\\
                       & \Pi & & 
    \end{array}
  \right)    $
\end{minipage}  
\end{eqnarray*}
}

\noindent and

{\small
\begin{eqnarray*}
\hspace{-6cm}
\begin{minipage}{0.5\textwidth}
    $\prod_4'$    
\end{minipage} & \hspace{-6.8cm}= &
\begin{minipage}{0.5\textwidth}\vspace{-.1cm}
  $F'\left(
    \begin{array}{llcc}
      \multirow{5}{*}{
      \alwaysNoLine
      \def\extraVskip{2pt}
      \AxiomC{$\Seq{\to_3, \to_1, A, B}{B}$}
      \DisplayProof; } & \Gamma   &= &\{\to_3, \to_1, A\};\\
                       & \Delta   &= &\{\to_3, \to_1\};\\
                       & \Upsilon &= &\{\to_3^A, \to_1^A, \to_3^B, \to_1^B, A^B\};\\
                       & \Sigma   &= &\{B, A\};\\
                       & \Pi & & 
    \end{array}
  \right)  
  $
\end{minipage}  
\end{eqnarray*}
}

Finally, after finish the translation process we obtained the translated proof
of~\eqref{proof:lmt-translated}.

\begin{landscape}
  \thispagestyle{empty}
  \begin{adjustwidth}{-8em}{-6em}
    \vspace{20cm}
    {\small
      \begin{equation}
        \label{proof:lmt-translated}
        \infer[\imply\mathrm{r}]
        {\Seq{\{\}}{[], (((( A  \to_0  B ) \to_1  A ) \to_2  A ) \to_3  B ) \to_4 B}}
        {
          \infer[\imply\mathrm{focus}]
          {\Seq{\{\},\to_3}{[], B}}
          {
            \infer[\imply\mathrm{l}]
            {\Seq{\{\to_3\},\to_3}{[], B}}              
            {
              \infer[\imply\mathrm{r}]
              {\Seq{\{\to_3\}, \to_3^B, \to_3}{[B], \to_2}}                            
              {
                \infer[\imply\mathrm{focus}]
                {\Seq{\{\to_3\}, \to_3^B, \to_3, \to_1}{[B], A}}              
                {
                  \infer[\imply\mathrm{l}]
                  {\Seq{\{\to_3, \to_1\}, \to_3^B, \to_3, \to_1}{[B], A}} 
                  {
                    \deduce
                    {\Seq{\{\to_3, \to_1\}, \to_3^B, \to_3^A, \to_1^A, \to_3, \to_1}{[B, A], \to_0}}
                    {
                      \deduce
                      {\Pi'}
                      {
                        {
                          \infer[\imply\mathrm{r}]
                          {\Seq{\{\to_3, \to_1\}, \to_3^B, \to_3^A, \to_1^A, \to_1^B, A^B, \to_3^B, \to_3, \to_1, A}{[B, A], \to_2}}
                          {\Seq{\{\to_3, \to_1\}, \to_3^B, \to_3^A, \to_1^A, \to_1^B, A^B, \to_3^B, \to_3, \to_1, A}{[B, A], A}}
                        }
                        &
                        {\Seq{\{\to_3, \to_1\}, \to_3^B, \to_3^A, \to_1^A, \to_1^B, A^B, \to_3^B, \to_3, \to_1, A}{[B, A], A}}
                      }
                    }                            
                  } &
                  {\Seq{\{\to_3, \to_1\}, \to_3^B, \to_3, \to_1, A}{[B], A}}                                                       
                }
              }
              &                          
              {\Seq{\{\to_3\}, \to_3, B}{[], B}}}
          }                 
        }              
      \end{equation}
    }
    \end{adjustwidth}
\end{landscape}

\subsection{Termination}

To control the end of the proof search procedure of \lmt~our approach is to
define an upper bound limit to the size of its proof search tree. Then, we need
to show that the \lmt~strategy here proposed allows exploring all the possible
ways to expand the proof tree until it reaches this size.

From Theorem~\ref{theo:upperbound-lj}, we know that the upper bound for cut-free
proofs based on \ljarrow~is ${|\alpha| \cdot 2^{|\alpha| + 1}}$, where $\alpha$
is the initial formula that we want to prove. We use the translation presented
in~Figure~\ref{fig:lj2lmt} on the previous Section to find a similar limit
for~\lmt proofs. We have to analyze three cases to establish an upper bound for
\lmt.

The cases are described below and are summarized in Table~\ref{tab:size}.

\begin{enumerate}[(i)]
\item Axioms of \ljarrow~maps one to one with axioms of \lmt;
\item $\imply$-right applications of \ljarrow~maps one to one with
  $\imply$-right applications of \lmt;
\item $\imply$-left applications of \ljarrow~maps to \lmt~in three different
  possible sub-cases, according to the context $(\alpha \imply \beta, q)$ in
  which the rule is being applied in~\ljarrow. We have to consider the fragment
  of \lmt~already translated to decide the appropriate case.
  \bigskip
  \begin{itemize}
  \item If the context is \textbf{not yet focused neither expanded}
    \begin{blockquote}
      Then, one application of \imply-left in \ljarrow~maps to two applications
      of rules in \lmt: first, a focus application, then an \imply-left
      application.
    \end{blockquote}
    
  \item If the context is \textbf{already focused but not yet expanded}

    \begin{blockquote}
      Then, one application of \imply-left in \ljarrow~maps to one application
      of \imply-left in \lmt.
    \end{blockquote}

  \item If the context $(\alpha \imply \beta, q)$ is \textbf{already focused and
      expanded}

    \begin{blockquote}
      Then, the one application of \imply-left in \ljarrow~maps to the height of
      the \lmt~proof fragment produced by the execution of the $PROOFUNTIL$
      function. Let this height be called $h$.
    \end{blockquote}

  \end{itemize}
\end{enumerate}

\begin{table}[h]
\centering
\begin{tabular}{c|cc|c|l}
\hline
\ljarrow                     & \multicolumn{2}{c}{\lmt}         & Map &                               \\
\cline{1-4}  
axiom                        & \multicolumn{2}{c}{axiom}        & 1:1 &                               \\
\cline{1-4}
\imply-right                 & \multicolumn{2}{c}{\imply-right} & 1:1 &                               \\
\hline
                             & focused         & expanded       &     &                               \\
\multirow{3}{*}{\imply-left} & No              & No             & 1:2 & 1 focus and 1 \imply-left     \\
                             & Yes             & No             & 1:1 & 1 \imply-left                 \\
                             & Yes             & Yes            & 1:h & 1 to the size of $PROOFUNTIL$ \\
\hline                                                                        
\end{tabular}
\caption{Mapping the size of \ljarrow~proofs into \lmt}
\label{tab:size}
\end{table}

\bigskip\bigskip

\begin{lemma}
  \label{lemma:size-h}
  The height $h$ that defines the size of the proof fragment returned by the
  function $PROOFUNTIL$ has a maximum limit of $2^{2log_2|\alpha|}$, where
  $\alpha$ is the main formula of the initial sequent of the proof in \lmt.
\end{lemma}

\begin{proof}
  Consider a proof $\prod_\ljarrow$ of an initial sequent in \ljarrow~with the
  form $\Seq{}{\alpha}$. The process of translating $\prod_\ljarrow$ to
  \lmt~produces a proof $\prod_\lmt$ with the initial sequent in the form
  $\Seq{\{\}}{[],\alpha}$. Consider that $\alpha$ has the form $\alpha_1 \imply
  \alpha_2$. In some point of the translation to \lmt, we reach a point where a
  context $(\psi \imply \varphi, q)$ is already focused and expanded in the
  already translated part of the proof $\prod_\lmt$. At this point, the function
  $PROOFUNTIL$ generates a fragment of the proof $\prod_\lmt$, call it $\Sigma$
  of size $h$. The height $h$ is bound by the number of applications of
  \imply-left rules in $\Sigma$, which can be determined by the multiplication of
  the degree of the formula $\alpha_1$ (that bounds the number of possible
  implicational formulas in the left side of a sequent in \lmt) by the maximum
  number of atomic formulas ($n$) inside the []-area in the highest branch of
  $\Sigma$ (each $p_i$ inside the []-area allows one application of the restart
  rule). Thus we can formalize this in the following manner:

  \begin{equation*}
    \begin{array}{lcl}
      h & = & n \times |\alpha_1| \\
      h & = & |\alpha| \times |\alpha| \\
      h & = & |\alpha|^2
    \end{array}
  \end{equation*}

Since
  \begin{equation*}
    \begin{array}{lcl}
      |\alpha|^2 & = & 2^{2 \cdot log_2|\alpha|}
    \end{array}
  \end{equation*}

  Then, we have that
  \begin{equation*}
    \begin{array}{lcl}
      h & = & 2^{2 \cdot log_2|\alpha|}
    \end{array}
  \end{equation*}  
\end{proof}

\bigskip

\begin{theorem}
  \label{theo:upperbound-lmt}
  Let $\alpha$ be a \mil~tautology. The size of any proof of $\alpha$ generated
  by  $\mathcal{S}$-\lmt  has an upper bound of $|\alpha| \cdot 2^{|\alpha| + 1 + 2 \cdot log_2|\alpha|}$.
\end{theorem}

\begin{proof}
  Considering the size of proofs for a formula $\alpha$ using \ljarrow, the
  mapping in Table~\ref{tab:size} and the Lemma~\ref{lemma:size-h}, the proof
  follows directly.
\end{proof}

\begin{theorem}
\label{theo:terminate} 
$\mathcal{S}$-\lmt~always terminates.   
\end{theorem}

\begin{proof}
  To guarantee termination, we use the upper bound presented in
  Theorem~\ref{theo:upperbound-lmt} to limit the height of opened branches
  during the \lmt~proof search process. The strategy presented in
  Section~\ref{sec:strategy} forces an ordered application of rules that
  produces all possible combinations of formulas to be expanded in the right and
  left sides of generated sequents. In other words, when the upper bound is
  reached, the proof search had, for sure, applied all possible expansion. The
  procedure can only continue by applying an already done expansion. 
   
\begin{itemize}
\item $\imply$-right rule is applied until we obtain an atomic formula on the
  right side.
\item focus rule is applied until every non-labeled formula becomes focused. The
  same formula can not be focused twice unless a restart rule is applied.
\item $\imply$-left rule can not be applied more than once to the same context
  unless a restart rule is applied.
\item between two applications of the restart rule in a branch there is only one
  possible application of a $\imply$-left rule for a context $(\alpha, q)$.
  $\alpha$ and $q$ are always subformulas of the initial formula.
\item restart rule is applied for each atomic formula that appears on the right
  side of sequents in a branch in the order of its appearance in the []-area,
  which means that proof search will apply the restart rule for each $p_i, i =
  1, \ldots n$ until the branch reaches the defined limit.
\end{itemize}

\end{proof}

The proof of completeness of \lmt~is closely related to this strategy and with
the way the proof tree is labeled during the proving process. Section
\ref{sec:soundness} presents the soundness proof of \lmt~and Section
\ref{sec:completeness}, the completeness proof.
 
%%%%%%%%%%%%%%%%%%%%%%%%%%%%%%%%%%%%%%%%%%%%%%%%%%%%%%%%%%%%%%%%%%%%%%%%%%%%%%%%%%%%%%%%%%%%%%%%%%%%%%%%%%%%%%%%%%%%%%%%%%%%
\subsection{Soundness} 
\label{sec:soundness}

In this section, we prove the soundness of \lmt. A few basic facts and
definitions used in the proof follow.

\begin{definition}
  \label{def:valid}
  A sequent \Seq{\{\Delta'\}, \Upsilon_1^{p_1}, \Upsilon_2^{p_2}, \ldots,
    \Upsilon_n^{p_n}, \Delta}{[p_1, p_2, \ldots, p_n], \varphi} is \emph{valid},
  if and only if, \todo[color=yellow!100]{ch04-lmt - Doubt: One of the members
    of Dedukti argued that as the $\Upsilon$ sets includes the previous ones we
    do not need this Union in the notation}$\Delta', \Delta \models \varphi$ or
  $\exists i \displaystyle (\bigcup_{k=1}^{i}{\Upsilon_k^{p_k}}) \models p_i$,
  for $i=1, \ldots n$.
\end{definition}

\begin{definition}
  We say that a rule is \emph{sound}, if and only if, in the case of the
  premises of the sequent are valid sequents, then its conclusion also is.
\end{definition}

A calculus is sound, if and only if, each of its rules is sound. We prove the
soundness of \lmt~by showing that this is the case for each one of its rules.
\\
\begin{proposition}\label{prop:sound}
  Considering validity of a sequent as defined in Definition~\ref{def:valid},
  \lmt~is sound.
\end{proposition}

\begin{proof}
  We show that supposing that premises of a rule are valid then, the validity of
  the conclusion follows. In the sequel, we analyze each rule of \lmt.
 
  \textbf{\imply-left} We need to analyze both premises together. Thus we have
  the combinations described below.
  \begin{itemize}
  \item Supposing the left premise is valid because $\alpha\imply\beta, \Delta',
    \Delta \models \alpha$ and the right premise is valid because
    $\alpha\imply\beta, \Delta', \Delta, \beta \models q$. We also know that
    $\alpha \imply \beta \in \Delta$ and $\Delta' \subseteq \Delta$. In this
    case, the conclusion holds:
      
        \begin{displaymath}      
          \AxiomC{$\alpha\imply\beta\;\Delta'\;\Delta$}
          \noLine
          \UnaryInfC{$\prod$}
          \noLine
          \UnaryInfC{$\alpha$}
          \AxiomC{$\alpha\imply\beta$}
          \BinaryInfC{$\beta$}
          \UnaryInfC{$q$}
          \DisplayProof
        \end{displaymath}

    \item Supposing the left premise is valid because $\exists i \displaystyle (\bigcup_{k=1}^{i}{\Upsilon_k^{p_k}}) \models p_i$, for $i=1,\ldots,n$, the conclusion holds as it is the same. Supposing the left premise is true because $\Delta^q \models q$, the conclusion also holds, as $\Delta^q = \Delta$.

    \item Supposing the right premise is valid because $\exists i \displaystyle (\bigcup_{k=1}^{i}{\Upsilon_k^{p_k}}) \models p_i$, for $i=1,\ldots,n$, then conclusion also holds. 
    \end{itemize}

  \textbf{restart} Here, we have three cases to evaluate.
    \begin{itemize}
    \item Supposing the premise is valid because $\Upsilon_1^{p_1}, \Upsilon_2^{p_2}, \ldots, \Upsilon_i^{p_i} \models p_i $, then $\exists i \displaystyle (\bigcup_{k=1}^{i}{\Upsilon_k^{p_k}}) \models p_i$, for $i=1,\ldots,n$. The conclusion is also valid. 
    \item Supposing the premise is valid because $\exists j \displaystyle (\bigcup_{k=i+1}^{j}{\Upsilon_k^{p_k}}) \models p_j$, for $j=i+1,\ldots,n$, then conclusion also holds.
    \item Supposing the premise is valid because $\Delta^q \models q$, then $\Delta \models q$ and, as  $\Delta' \subseteq \Delta$, $\Delta', \Delta \models q$. 
    \end{itemize}

  \textbf{\imply-right} 
    \begin{itemize}
    \item Supposing the premise is valid because $\Upsilon_1^{p_1}, \Upsilon_2^{p_2}, \ldots, \Upsilon_i^{p_i} \models p_i $, then $\exists i \displaystyle (\bigcup_{k=1}^{i}{\Upsilon_k^{p_k}}) \models p_i$, for $i=1,\ldots,n$. This is also valid in the conclusion. 
    \item Supposing the premise is valid because $\Delta', \Delta, \alpha \models \beta$, then every Kripke model that satisfies $\Delta'$, $\Delta$ and $\alpha$ also satisfies $\beta$. We know that $\Delta' \subseteq \Delta$. Those models also satisfies $\alpha \imply \beta$ and, then, conclusion also holds.
    \end{itemize}

  \textbf{focus}
    \begin{itemize}
    \item Supposing the premise is valid because $\Upsilon_1^{p_1}, \Upsilon_2^{p_2}, \ldots, \Upsilon_i^{p_i} \models p_i $, then $\exists i \displaystyle (\bigcup_{k=1}^{i}{\Upsilon_k^{p_k}}) \models p_i$, for $i=1,\ldots,n$. This is also valid in the conclusion. 
    \item Supposing the premise is valid because $\Delta', \alpha, \Delta, \alpha \models \beta$, then the conclusion also holds as $\Delta', \Delta, \alpha \models \beta$.
    \end{itemize}
\end{proof}

From Proposition~\ref{prop:sound}, we conclude that \lmt~only prove tautologies. 

%%%%%%%%%%%%%%%%%%%%%%%%%%%%%%%%%%%%%%%%%%%%%%%%%%%%%%%%%%%%%%%%%%%%%%%%%%%%%%%%%%%%%%%%%%%%%%%%%%%%%%%%%%%%%%%%%%%%%%%%%%%%
\subsection{Completeness}
\label{sec:completeness}

By Observation~\ref{obs:strategy}.\ref{obs:strategy-2}~we know that a top sequent of an open branch in an attempt proof tree has the general form below, where $q$ is an atomic formula: 
\begin{equation*}
  \Seq{\{\Delta'\}, \Upsilon_1^{p_1}, \Upsilon_2^{p_2}, ..., \Upsilon_n^{p_n}, \Delta}{[p_1, p_2,..., p_n], q}  
\end{equation*}

From Definition~\ref{def:valid} and considering that $\Delta' \in \Delta$ in any sequent of an attempt proof tree following our proposed strategy, we can \todo[color=yellow!100]{ch04:lmt - Doubt: Jean-Pierre Jouannaud said that is not a definition, but a proposition}define a invalid sequent as follows:
\begin{definition}
\label{def:invalid}  
  A sequent is \emph{invalid} if and only if  $\Delta \nvDash q$ and $\forall i \displaystyle (\bigcup_{k=1}^{i}{\Upsilon_k^{p_k}}) \nvDash p_i$, for $i=1, \ldots, n$.     
\end{definition}

Our proof of completeness starts with a definition of atomic formulas on the left and right sides of a top sequent.

\begin{definition}
\label{def:countermodel}
We can construct a Kripke counter-model $\mathcal{M}$ that satisfies atomic formulas on the right side of a top sequent, and that falsifies the atomic formula on the left. This construction can be done in the following way:

\begin{enumerate}
\item\label{def:countermodel-1} The model $\mathcal{M}$ has an initial world $w_0$.

\item\label{def:countermodel-2} By the proof strategy, we can conclude that, in any sequent of the proof tree, $\Upsilon_1^{p_1} \subseteq \Upsilon_2^{p_2} \subseteq \cdots \subseteq \Upsilon_n^{p_n} \subseteq \Delta$. We create a world in the model $\mathcal{M}$ corresponding for each one of these bags of formulas and, using the inclusion relation between them, we define a respective accessibility relation in the model $\mathcal{M}$ between such worlds. That is, we create worlds $w_{\Upsilon_1^{p_1}}, w_{\Upsilon_2^{p_2}}, \ldots, w_{\Upsilon_n^{p_n}}, w_{\Delta}$ related in the following form: $w_{\Upsilon_1^{p_1}} \preceq w_{\Upsilon_2^{p_2}} \preceq \cdots \preceq w_{\Upsilon_n^{p_n}} \preceq w_{\Delta}$. As $w_0$ is the first world of $\mathcal{M}$, it precedes $w_{\Upsilon_1^{p_1}}$, that is, $w_0 \preceq w_{\Upsilon_1^{p_1}}$ is also included in the accessibility relation. If $\Upsilon_i^{p_i} = \Upsilon_{i+1}^{p_{i+1}}$, for $i = 1, \ldots, n$, then the associated worlds that correspond to those sets have to be collapsed in a single world $w_{\Upsilon_i^{p_i}-\Upsilon_{i+1}^{p_{i+1}}}$. In this case, the previous relation $w_{\Upsilon_i^{p_i}} \preceq w_{\Upsilon_{i+1}^{p_{i+1}}}$ is removed from the $\preceq$ relation of the model $\mathcal{M}$ and the pairs $w_{\Upsilon_{i-1}^{p_{i-1}}} \preceq w_{\Upsilon_i^{p_i}}$ and $w_{\Upsilon_{i+1}^{p_{i+1}}} \preceq w_{\Upsilon_{i+2}^{p_{i+2}}}$ become respectively $w_{\Upsilon_{i-1}^{p_{i-1}}} \preceq w_{\Upsilon_i^{p_i}-\Upsilon_{i+1}^{p_{i+1}}}$ and $w_{\Upsilon_i^{p_i}-\Upsilon_{i+1}^{p_{i+1}}} \preceq w_{\Upsilon_{i+2}^{p_{i+2}}}$.
  
\item\label{def:countermodel-3} By the Definition~\ref{def:invalid} of an invalid sequent, $\Delta \nvDash q$. The world $w_{\Delta}$ will be used to guarantee this. We set $q$ false in $w_{\Delta}$, i.e, $\mathcal{M} \nvDash_{w_{\Delta}} q$. We also set every atomic formula that is in $\Delta$ as true, i.e., $\forall p, p \in \Delta, \mathcal{M} \vDash_{w_{\Delta}} p$.
  
\item\label{def:countermodel-4} By the Definition~\ref{def:invalid} of an invalid sequent, we also need that $\forall i \displaystyle (\bigcup_{k=1}^{i}{\Upsilon_k^{p_k}}) \nvDash p_i$, for $i=1, \ldots n$. Thus, for each $i, i=1, \ldots, n$ we  set $\mathcal{M} \nvDash_{w_{{\Upsilon_i^{p_i}}}} p_i$ and $\forall p, p \in \Upsilon_i^{p_i}$, being $p$ an atomic formula, $\mathcal{M} \vDash_{w_{{\Upsilon_i^{p_i}}}} p$. In the case of collapsed worlds, we keep the satisfaction relation of the previous individual worlds in the collapsed one. 

\item\label{def:countermodel-5} In $w_0$ set every atomic formula inside the []-area (all of them are atomic) as false. That is, $\mathcal{M} \nvDash_{w_0} p_i$, for $i=1, \ldots, n$. We also set the atomic formula outside the []-area false in this world: $\mathcal{M} \nvDash_{w_0} q$. Those definitions make $w_0$ consistent with the $\preceq$ relation of $\mathcal{M}$.
\end{enumerate}
\end{definition}

The Figure~\ref{fig:countermodel} shows the general shape of counter-models following the steps enumerated above. This procedure to construct counter-model allows us to state the following lemma:

\begin{figure}[h]
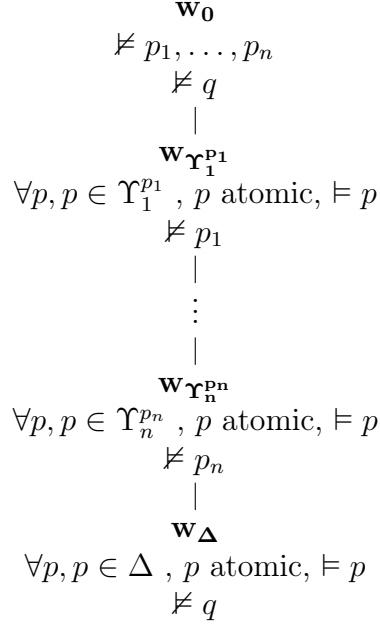

{
     \Tree [.{$\mathbf{w_0}$ \\ $\nvDash p_1,\ldots, p_n$ \\ $\nvDash q$} [.{$\mathbf{w_{{\Upsilon_1^{p_1}}}}$ \\ $\forall p, p \in \Upsilon_1^{p_1}$ , $p$ atomic, $\vDash p$ \\ $\nvDash p_1$} [.{$\vdots$} [.{$\mathbf{w_{{\Upsilon_n^{p_n}}}}$ \\ $\forall p, p \in \Upsilon_n^{p_n}$ , $p$ atomic, $\vDash p$ \\ $\nvDash p_n$} [.{$\mathbf{w_\Delta}$ \\ $\forall p, p \in \Delta$ , $p$ atomic, $\vDash p$ \\ $\nvDash q$} ]]]]]    
}
\caption{General schema of counter-models}
\label{fig:countermodel}
\end{figure}

\begin{lemma}
\label{lemma:countermodel}
Let $S$ be a top sequent of an open branch in an attempt proof tree generated by the strategy presented in Section~\ref{sec:strategy}. Then we can construct a Kripke model $\mathcal{M}$ with a world $u$ where $\mathcal{M} \nvDash_u S$, using the proposed counter-model generation procedure.
\end{lemma}

\begin{proof}
We can prove this by induction on the degree of formulas in $\Delta$. From Definition~\ref{def:countermodel}, items~\ref{def:countermodel-3} and~\ref{def:countermodel-4} we know the value of each atomic formula in the worlds $w_\Delta$ and in each world $w_{\Upsilon_i^{p_i}}$. The inductive hypothesis is that every formula in $\Delta$ is true in $w_\Delta$. Thus, as $\Upsilon_1^{p_1} \subseteq \Upsilon_2^{p_2} \subseteq \cdots \subseteq \Upsilon_n^{p_n} \subseteq \Delta$, every formula in $\Upsilon_i^{p_i}$ is true in $w_{\Upsilon_i^{p_i}}$, for $i=1, \ldots, n$.

Thus, we have two cases to consider:

\begin{enumerate}
\item The top sequent is in the rightmost branch of the proof tree ([]-area is empty).
\label{lemma:countermodel-1}
\\
Let $\alpha\imply\beta$ be a formula in \mil~that is in $\Delta$. We show that $\mathcal{M} \vDash_{w_\Delta} \alpha\imply\beta$. In this case, by the proof strategy, $\beta \equiv (\beta_1 \imply (\beta_2 \imply \cdots \imply (\beta_m \imply p)))$, where $p$ is an atomic formula. By Definition~\ref{def:countermodel}.\ref{def:countermodel-3} $\vDash_{w_\Delta} p$. As $w_\Delta$ has no accessible world from it (except for itself), $\vDash_{w_\Delta} \beta$. By the proof strategy, $\beta_m \imply p, \beta_{m-1} \imply \beta_m \imply p, \ldots, \beta_2 \imply \cdots \imply \beta_{m-1} \imply \beta_m \imply p, \beta_1 \imply \beta_2 \imply \cdots \imply \beta_{m-1} \imply \beta_m \imply p$ also are in $\Delta$. The degree of each of these formulas is less than the degree of $\alpha\imply\beta$ and, by the induction hypothesis, all of them are true in $w_\Delta$. Thus $\vDash_{w_\Delta} \beta$ and $\vDash_{w_\Delta} \alpha\imply\beta$.
\\
As the []-area is empty, the sets $\Upsilon_i^{p_i}$ are also empty. The counter-model only has two words, $w_0$ and $w_\Delta$, following the properties described in Definition~\ref{def:countermodel}.

\item The top sequent is in any other branch that is not the rightmost one ([]-area is not empty).
\label{lemma:countermodel-2}
\\
Let $\alpha\imply\beta$ be a formula in \mil~that is in $\Delta$. We show that $\mathcal{M} \vDash_{w_\Delta} \alpha\imply\beta$. In this case, by the proof strategy, $\alpha \equiv (\alpha_1 \imply (\alpha_2 \imply \cdots \imply (\alpha_m \imply q)))$, where $q$ is the atomic formula in the right side of the sequent, out of the []-area. By Definition~\ref{def:countermodel}.\ref{def:countermodel-3} $\nvDash_{w_\Delta} q$. By the proof strategy, $\alpha_1, \alpha_2, \dots, \alpha_m$ also are in $\Delta$. The degree of each of these formulas is less than the degree of $\alpha\imply\beta$ and, by the induction hypothesis, all of them are true in $w_\Delta$. This ensures $\nvDash_{w_\Delta} \alpha$ and $\vDash_{w_\Delta} \alpha\imply\beta$.
\\
Considering now a formula $\alpha\imply\beta$ from \mil~that is in $\Upsilon_i^{p_i}$. By Definition~\ref{def:countermodel}.\ref{def:countermodel-2}, $\alpha\imply\beta$ also belongs to $\Delta$. From the last paragraph, we show that, for any formula $\alpha\imply\beta \in \Delta$, $\nvDash_{w_\Delta} \alpha$. As $\nvDash_{w_\Delta} \alpha$, by the accessibility relation of the Kripke model, $\nvDash_{w_{\Upsilon_i^{p_i}}} \alpha$, for each $i=1, \ldots, n$. Thus, the value of $\alpha \imply \beta$ is defined in any of these worlds by the value of $\alpha \imply \beta$ in $w_\Delta$, that we showed to be true. Thus, $\vDash_{w_{\Upsilon_i^{p_i}}} \alpha\imply\beta$. 
   
\end{enumerate}

As stated in Definition~\ref{def:countermodel}.\ref{def:countermodel-2}, $\Upsilon_1^{p_1} \subseteq \Upsilon_2^{p_2} \subseteq \cdots \subseteq \Upsilon_n^{p_n} \subseteq \Delta$ and following the accessibility relation rule of the \mil~semantic (relations are reflexive and transitive) we conclude that:

\begin{align*}
  \mathcal{M} &\vDash_{w_0} \Upsilon_1^{p_1}, \nvDash_{w_0} p_1 \\
              &\vDash_{w_0} \Upsilon_2^{p_2}, \nvDash_{w_0} p_2 \\
              & \vdots \\
              &\vDash_{w_0} \Upsilon_n^{p_n}, \nvDash_{w_0} p_n \\
              &\vDash_{w_0} \Delta, \nvDash_{w_0} q \\
\end{align*}

Proving Lemma~\ref{lemma:countermodel}.
\end{proof}

\bigskip

\begin{definition}\label{def:double}
% A rule is said \emph{invertible} or \emph{double-sound} if it is not only truth-preserving from the premises to conclusion but also from the conclusion to each of its premises.  

  A rule is said \emph{invertible} or \emph{double-sound} iff the validity of
  its conclusion implies the validity of its premises.
\end{definition}

In other words, by Definition~\ref{def:double}, we know that a counter-model for
a top sequent of a proof tree that can not be expanded anymore can be used to
construct a counter-model to every sequent in the same branch of the tree until
the conclusion (root sequent). In the case of the \imply-right rule in our
system, not just if the premise of the rule has a counter-model, then so does the
conclusion, but the same counter-model will do. \citet{weich1998} called rules
with this property \emph{preserving counter model}. Dyckhoff (personal
communication, 2015) proposed to call this kind of rules of \emph{strongly
  invertible} rules. In the case of the \imply-left rule, this is the same when
one of the premises is valid, but, considering the case that both premises are
not valid, we need to mix the counter-models of both sides to construct the
counter-model for the conclusion of the rule. This way to produce counter-models
is what we call a \emph{weakly invertible} rule.

\begin{lemma}
\label{lemma:invertible}
  The rules of \lmt are invertible. 
\end{lemma}

\begin{proof}
  We show that the rules of \lmt~are invertible when considering a proof tree
  labeled in the schema presented in Section~\ref{sec:strategy}. We prove that
  for the structural rules (focus and restart) and \imply-right, from the
  existence of a Kripke model that makes the premise of the rule invalid follows
  that the conclusion is also invalid. For the \imply-left rule, from the Kripke
  models of the premises, we can construct a Kripke model that also makes the
  conclusion of the rule invalid.

\textbf{\imply-right} If the premise is invalid, then there is a Kripke model $\mathcal{M}$ where $\Delta', \Delta, \alpha \nvDash \beta$ and $\forall i \displaystyle (\bigcup_{k=1}^{i}{\Upsilon_k^{p_k}}) \nvDash p_i$, for $i=1, \ldots, n$ in a given world $u$ of $\mathcal{M}$. Thus, in the conclusion we have:
\begin{itemize}
\item By the definition of semantics of Section~\ref{sec:mimp-semantic}, there have to be a world $v$, $u \preceq v$, in the model $\mathcal{M}$ where $\Delta', \Delta, \alpha$ are satisfied and where $\beta$ is not. Thus, in $u$, $\alpha\imply\beta$ can not hold.
\item By the model $\mathcal{M}$, for each $i$, exists a world $v_i$, $u \preceq v_i$, where $\vDash_{v_i} \Upsilon_i^{p_i}$ and $\nvDash_{v_i} p_i$.
\item Thus, the conclusion is also invalid.
\end{itemize}

\textbf{\imply-left} Considering that one of the premises of \imply-left is not valid, the conclusion also is. We have to evaluate three cases:
\begin{enumerate}
  
\item \textbf{The right premise is invalid but the left premise is valid}. Then there is a Kripke model $\mathcal{M}$ where $\alpha \imply \beta, \Delta', \Delta, \beta \nvDash q$ and $\forall i \displaystyle (\bigcup_{k=1}^{i}{\Upsilon_k^{p_k}}) \nvDash p_i$, for $i=1, \ldots, n$ from a given world $u$. Thus, in the conclusion we have: 
  \begin{itemize}
  \item By the model $\mathcal{M}$, there have to be a world $v$, $u \preceq v$, in the model where $\alpha \imply \beta, \Delta', \Delta, \beta$ are satisfied and where $q$ is not.
  \item By the model $\mathcal{M}$, for each $i$, exists a world $v_i$, $u \preceq v_i$, where $\vDash_{v_i} \Upsilon_i^{p_i}$ and $\nvDash_{v_i} p_i$.
  \item Thus, the conclusion is invalid too.  
  \end{itemize}

\item \textbf{The left premise is invalid but the right premise is valid}. Then there is a Kripke model $\mathcal{M}$ where $\alpha \imply \beta, \Delta', \Delta \nvDash \alpha$ and $\forall i \displaystyle (\bigcup_{k=1}^{i}{\Upsilon_k^{p_k}}) \nvDash p_i$, for $i=1, \ldots, n$, and $\Delta^q \nvDash q$ from a given world $u$. Thus, in the conclusion we have: 
  \begin{itemize}
  \item By the model $\mathcal{M}$, there have to be a world $v$, $u \preceq v$, in the model where $\alpha \imply \beta, \Delta', \Delta$ are satisfied and where $\alpha$ is not. 
  \item By the model $\mathcal{M}$, for each $i$, exists a world $v_i$, $u \preceq v_i$, where $\vDash_{v_i} \Upsilon_i^{p_i}$ and $\nvDash_{v_i} p_i$.
  \item We also know by $\mathcal{M}$ that there is a world $v_{\Delta^q}$, $u \preceq v_{\Delta^q}$, where $\vDash_{v_{\Delta^q}} \Delta^{q}$ and $\nvDash_{v_{\Delta^q}} q$. We also have that $\Delta^q = \Delta$ and that $\alpha \imply \beta \in \Delta$. Therefore, $\vDash_{v_{\Delta^q}} \Delta'$ and $\vDash_{v_{\Delta^q}} \alpha \imply \beta$.
  \item Thus, the conclusion can not be valid.    
  \end{itemize}

\item \textbf{Both left and right premises are invalid}. Then there are two models $\mathcal{M}_1$ and $\mathcal{M}_2$, from the right and left premises respectively. In $\mathcal{M}_1$ there is a world $u_1$ that makes the right sequent invalid as described in item 1. In $\mathcal{M}_2$ there is a world $u_2$ that makes the sequent of the left premise invalid as described in item 2. Considering the way Kripke models are constructed based on Lemma~\ref{lemma:countermodel}, we know that $u_1$ and $u_2$ are root worlds of their respective counter-models. Thus, converting the two models into one, $\mathcal{M}_3$,  by mixing $u_1$ and $u_2$ in the root of $\mathcal{M}_3$, called $u_3$, we have that in $u_3$:
  \begin{itemize}
  \item $\alpha \imply \beta, \Delta', \Delta$ are satisfied and $\alpha$ is not.
  \item for $i=1, \ldots, n$, we have that $\vDash_{u_3} \Upsilon_i^{p_i}$ and $\nvDash_{u_3} p_i$.
  \item $\nvDash_{u_3} q$
  \item Thus, the conclusion is also invalid. 
  \end{itemize}
\end{enumerate}

\textbf{focus} If we have a model that invalidates the premise, this model also invalidates the conclusion as the sequents in the premise and in the conclusion are the same despite the repetition of the focused formula $\alpha$.
 
\textbf{restart} If the restart premise is invalid, then there is a Kripke model $\mathcal{M}$ and a world $u$ from which $\Upsilon_1, \Upsilon_2, \ldots, \Upsilon_i \nvDash p_i$ and $\forall j \displaystyle (\bigcup_{k=1}^{j}{\Upsilon_k^{p_k}}) \models p_j$, for $j = i+1, \ldots, n$, and $\Delta^q \nvDash q$. Thus, in the conclusion we have:
\begin{itemize}
  \item By the model $\mathcal{M}$, there have to be a world $v$, $u \preceq v$, in the model where $\Upsilon_1, \Upsilon_2, \ldots, \Upsilon_i$ are satisfied and where $p_i$ is not. Each $\Upsilon_k$ has the same formulas as $\Upsilon_k^{p_k}$ and, by the restart condition, we know that $\nvDash p_k$, for $k=1, \ldots, i$.
  \item By the model $\mathcal{M}$, for each $j$, exists a world $v_j$, $u \preceq v_j$, where $\vDash_{v_j} \Upsilon_j^{p_j}$ and $\nvDash_{v_j} p_j$.
  \item We also know by $\mathcal{M}$ that there is a world $v_q$, $u \preceq v_q$, where $\vDash_{v_q} \Delta^{q}$ and $\nvDash_{v_q} q$. We also have that $\Delta' \subseteq \Delta$. Therefore, $\vDash_{v_q} \Delta'$.
  \item Thus, the conclusion is invalid.   
\end{itemize}
\end{proof}

Now we can state a proposition about completeness of~\lmt:

\begin{proposition}
  \lmt~is complete regarding the proof strategy presented in Section~\ref{sec:strategy}
\end{proposition}

\begin{proof}
  It follows direct from Proposition~\ref{theo:terminate} (\lmt~terminates) and
  Lemma~\ref{lemma:countermodel} (we can construct a counter-model for a top
  sequent in a terminated open branch of \lmt) and Lemma~\ref{lemma:invertible}
  (the rules of \lmt~are invertible).
\end{proof}

%%%%%%%%%%%%%%%%%%%%%%%%%%%%%%%%%%%%%%%%%%%%%%%%%%%%%%%%%%%%%%%%%%%%%%%%%%%%%%%%%%%%%%%%%%%%%%%%%%%%%%%%%%%%%%%%%%%%%%%%%%%%
\subsection{Examples}

\begin{example}  
  As an example, consider the Peirce formula, ${((A \to B) \to A) \to A}$, which
  is very known to be a Classic tautology, but not a tautology in Intuitionistic
  neither in Minimal Logic (which includes \mil). Proof~\eqref{proof:peirce}
  below shows the open branch of an attempt proof tree for this formula.
  Following our termination criteria, this branch should be higher than the
  fragment below, but, to help improve understanding, we stop that branch in the
  point from which proof search just produces repetition (similar to the use of
  a loop checker).

\begin{equation}
\label{proof:peirce}
{\footnotesize  
\infer[\to\mbox{right}-\to_{2}]
{\color{black}{{\color{black}\{\}} {\color{black}\Rightarrow_{0}} (((A\to_{0}B)\to_{1}A)\to_{2}A),{\color{black}[]}}}{
\infer[\mbox{focus}-\to_{1}]
{\color{black}{{\color{black}\{\}},\to_{1} {\color{black}\Rightarrow_{1}} {\color{black}[]},A}}{
\infer[\to\mbox{left}-\to_{1}]
{\color{black}{\to_{1},{\color{black}\{}\to_{1}{\color{black}\}} {\color{black}\Rightarrow_{2}} {\color{black}[]},A}}{
\infer[\to\mbox{right}-\to_{0}]
{\color{black}{\to_{1},{\color{black}\{}\to_{1}{\color{black}\}},(\to_{1})^{A} {\color{black}\Rightarrow_{3}} {\color{black}[}A{\color{black}]},\to_{0}}}{
\infer[\mbox{focus}-A]
{\color{black}{\to_{1},{\color{black}\{}\to_{1}{\color{black}\}},(\to_{1})^{A},A {\color{black}\Rightarrow_{5}} {\color{black}[}A{\color{black}]},B}}{
\infer[\to\mbox{left}-\to_{1}]
{\color{black}{\to_{1},(\to_{1})^{A},A,{\color{black}\{}\to_{1},A{\color{black}\}} {\color{black}\Rightarrow_{6}} {\color{black}[}A{\color{black}]},B}}{
\vdots & \infer[\mbox{restart}-A]
{\color{black}{\to_{1},(\to_{1})^{A},A,{\color{black}\{}\to_{1},A{\color{black}\}} {\color{black}\Rightarrow_{9}} {\color{black}[}A{\color{black}]},B}}{
\infer[\mbox{focus}-\to_{1}]
{\color{black}{(\to_{1})^{B},\to_{1},(A)^{B},{\color{black}\{\}} {\color{black}\Rightarrow_{11}} {\color{black}[}B{\color{black}]},A}}{
\infer[\to\mbox{left}-\to_{1}]
{\color{black}{(\to_{1})^{B},\to_{1},(A)^{B},{\color{black}\{}\to_{1}{\color{black}\}} {\color{black}\Rightarrow_{12}} {\color{black}[}B{\color{black}]},A}}{
\infer[\to\mbox{right}-\to_{0}]
{\color{black}{(\to_{1})^{B},\to_{1},(A)^{B},{\color{black}\{}\to_{1}{\color{black}\}},(\to_{1})^{A} {\color{black}\Rightarrow_{13}} {\color{black}[}B,A{\color{black}]},\to_{0}}}{
\infer[\mbox{focus}-A]
{\color{black}{(\to_{1})^{B},\to_{1},(A)^{B},{\color{black}\{}\to_{1}{\color{black}\}},(\to_{1})^{A},A {\color{black}\Rightarrow_{16}} {\color{black}[}B,A{\color{black}]},B}}{
\infer[\to\mbox{left}-\to_{1}]
{\color{black}{(\to_{1})^{B},\to_{1},(A)^{B},(\to_{1})^{A},A,{\color{black}\{}\to_{1},A{\color{black}\}} {\color{black}\Rightarrow_{19}} {\color{black}[}B,A{\color{black}]},B}}{
\vdots & \infer[\mbox{restart}-B]
{\color{black}{(\to_{1})^{B},\to_{1},(A)^{B},(\to_{1})^{A},A,{\color{black}\{}\to_{1},A{\color{black}\}} {\color{black}\Rightarrow_{21}} {\color{black}[}B,A{\color{black}]},B}}{
\infer[\mbox{focus}-\to_{1}]
{\color{black}{\to_{1},A,(\to_{1})^{A},{\color{black}\{\}} {\color{black}\Rightarrow_{22}} {\color{black}[}A,B{\color{black}]},B}}{
\infer[\mbox{focus}-A]
{\color{black}{\to_{1},A,(\to_{1})^{A},{\color{black}\{}\to_{1}{\color{black}\}} {\color{black}\Rightarrow_{25}} {\color{black}[}A,B{\color{black}]},B}}{
\infer[\to\mbox{left}-\to_{1}]
{\color{black}{\to_{1},A,(\to_{1})^{A},{\color{black}\{}\to_{1},A{\color{black}\}} {\color{black}\Rightarrow_{26}} {\color{black}[}A,B{\color{black}]},B}}{
\vdots & {\color{black}{\to_{1},A,(\to_{1})^{A},{\color{black}\{}\to_{1},A{\color{black}\}} {\color{black}\Rightarrow_{28}} {\color{black}[}A,B{\color{black}]},B}}}
}
}
}
}
}
}
 & \vdots
}
}
}
}
}
}
 & \vdots
}
}
}
}
\end{equation}

\bigskip\bigskip

From the top sequent of the open branch ($\Rightarrow_{28}$) we can generalize the sequent as:

\begin{displaymath}
  \begin{array}{cccccc}
    \Delta' & \Upsilon_1^{p_1} & \Delta & \Rightarrow_{28} & [p_1, p_2] & q \\
    \{ \imply1, A\} & \imply_1^A & \imply_1, A & \Rightarrow_{28} & [A, B], & B
  \end{array}    
\end{displaymath}

\bigskip

Thus, following the method described in Section~\ref{sec:completeness} we can
extract the counter-model below that falsifies the sequent:

\begin{equation}
\label{model:peirce-from-lmt}
{
\Tree [.{$\mathbf{w_0}$ \\ $\nvDash A$ \\ $\nvDash B$ \\ $\nvDash A \imply_0 B$ \\ $\vDash (A \imply_0 B) \imply_1 A$} [.{$\mathbf{w_{{\Upsilon^A}}}$ \\ $\nvDash A$ \\ $\nvDash B$ \\ $\nvDash A \imply_0 B$ \\ $\vDash (A \imply_0 B) \imply_1 A$} [.{$\mathbf{w_\Delta}$ \\ $\vDash A$ \\ $\nvDash B$ \\ $\nvDash A \imply_0 B$ \\ $\vDash (A \imply_0 B) \imply_1 A$} ]]] 
}
\end{equation}

\bigskip

From Lemma~\ref{lemma:invertible}, we can extend this counter-model to falsify
the initial sequent ($\Rightarrow_{0}$), showing that the Pierce rule does not
hold on \mil.
\end{example}

\begin{example}
  As another example, we can consider the Dummett formula: $(A \imply B) \lor (B
  \imply A)$. It is known that a Kripke counter-model that falsifies this
  formula needs at least two branches in \ipl~and \mpl, so it is also in \mil.
  This example allows us to understand how to mix the right and left premises
  counter-models to falsifies a conclusion sequent of a \imply-left rule.

  As we want to use \lmt, we need to convert the Dummett formula from the form
  above to its implicational version. We use here the general translation
  presented in~\cite{haeusler2014a}. Thus, the translated version is a formula
  $\alpha$ as follows:

\bigskip

{\small
$\alpha \equiv ({((A \imply B) \imply A)} \imply {(((B \imply A) \imply A) \imply (C \imply A))}) \imply {(({((A \imply B) \imply B)} \imply {(((B \imply A) \imply B) \imply (C \imply B))}) \imply C)}$
}

\bigskip

To shorten the presentation, consider the following abbreviations:

\begin{displaymath}
\begin{array}{lll}
  \alpha_1 & = & ((A \imply B) \imply A) \imply (((B \imply A) \imply A) \imply (C \imply A)) \\
  \alpha_2 & = & ((A \imply B) \imply B) \imply (((B \imply A) \imply B) \imply (C \imply B))
\end{array}
\end{displaymath}

\bigskip

The tree~\eqref{proof:counter-model-gen-dummet}, presents a shortened version of
a completely expanded attempt proof tree in \lmt for the Dummet formula in
implicational form. We concentrate in show the application of the operational
rules ({\imply-left} and {\imply-right}). We removed from the tree the focus
area on the left of each sequent and the applications of structural rules. We
also exclude from the tree the labeled formulas, just showing the necessary
formulas for the specific point in the tree. The list of sequents above the
boxed sequents in the tree represents the top sequents of each branch after the total
expansion of $\alpha_2$.

\begin{landscape}
  \thispagestyle{empty}
  \begin{adjustwidth}{-5em}{-5em}
    \begin{gather}
      \AxiomC{$\overwrite{red}{\Seq{\ldots}{[C, A], A~(M_1)}}{$\Seq{\ldots}{[B,C], A}$ $\cancel{\Seq{\ldots,B}{[C], B}}$  $\Seq{\ldots}{[C,B],B}$  $\Seq{\ldots,A}{[C], B}$  $\Seq{\ldots}{[C], C}$  $\Seq{\ldots,B}{[ ], C}$}$}  
      \AxiomC{$\overwrite{red}{\Seq{\ldots, B}{[C], A~(M_2)}}{$\Seq{\ldots,B}{[B,C], A}$ $\cancel{\Seq{\ldots,B}{[C], B}}$  $\cancel{\Seq{\ldots,B}{[C,B],B}}$  $\cancel{\Seq{\ldots, B, A}{[C], B}}$  $\Seq{\ldots,B}{[C], C}$  $\Seq{\ldots,B}{[ ], C}$}$}
      \BinaryInfC{$\Seq{\ldots, A \imply B}{[C], A}$} 
      \UnaryInfC{$\Seq{\ldots, }{[C], (A \imply B) \imply A}$}
      \AxiomC{$\overwrite{red}{\Seq{\ldots}{[C, A], q~(M_1)}}{$\Seq{\ldots}{[B,C], A}$ $\cancel{\Seq{\ldots,B}{[C], B}}$  $\Seq{\ldots}{[C,B],B}$  $\Seq{\ldots,A}{[C], B}$  $\Seq{\ldots}{[C], C}$  $\Seq{\ldots,B}{[ ], C}$}$}
      \AxiomC{$\cancel{\Seq{\ldots, A}{[C], A}}$} 
      \BinaryInfC{$\Seq{\ldots, B \imply A}{[C], A}$}
      \UnaryInfC{$\Seq{\ldots, }{[C], (B \imply A) \imply A}$}
      \AxiomC{$\overwrite{red}{\Seq{\ldots, }{[C], C~(M_1)}}{$\Seq{\ldots}{[B,C], A}$ $\cancel{\Seq{\ldots,B}{[C], B}}$  $\Seq{\ldots}{[C,B],B}$  $\Seq{\ldots,A}{[C], B}$  $\Seq{\ldots}{[C], C}$  $\Seq{\ldots,B}{[ ], C}$}$}
      \AxiomC{$\overwrite{red}{\Seq{\ldots, A}{[ ], C~(M_3)}}{$\cancel{\Seq{\ldots,A}{[B,C], A}}$ $\cancel{\Seq{\ldots,A,B}{[C], B}}$  $\Seq{\ldots,A}{[C,B],B}$  $\Seq{\ldots,A}{[C], B}$  $\Seq{\ldots,A}{[C], C}$  $\Seq{\ldots,A,B}{[ ], C}$}$} 
      \BinaryInfC{$\Seq{\ldots, C \imply A}{[ ], C}$}                                                                                   
      \BinaryInfC{$\Seq{\ldots, ((B \imply A) \imply A) \imply (C \imply A)}{[ ], C}$}    
      \RightLabel{$\imply-\text{left}_{(\alpha_1, C)}$}
      \BinaryInfC{$\Seq{((A \imply B) \imply A) \imply (((B \imply A) \imply A) \imply (C \imply A)), \alpha_2}{[], C}$}
      \UnaryInfC{$\Seq{\alpha_1, \alpha_2}{[], C}$}          
      \UnaryInfC{$\Seq{}{[], \alpha_1 \imply (\alpha_2 \imply C)}$}          
      \UnaryInfC{$\Seq{}{[], \alpha}$}       
      \DisplayProof
      \label{proof:counter-model-gen-dummet}
    \end{gather}
    \end{adjustwidth}
\end{landscape}

From the open branches of the tree~\eqref{proof:counter-model-gen-dummet} we
extract the following three models. Models $M_1$ repeats in some branches. Thus,
we only represent it once here. We indicated in
tree~\eqref{proof:counter-model-gen-dummet}, in the top sequents of each branch,
the corresponding model generated on it, following the
Definition~\ref{def:countermodel}.

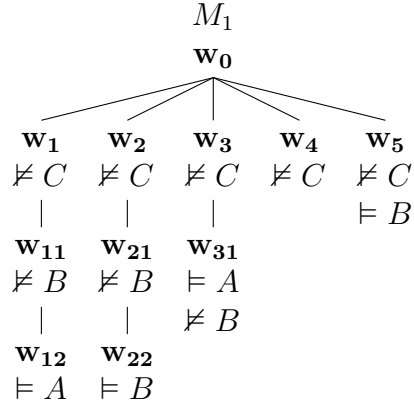
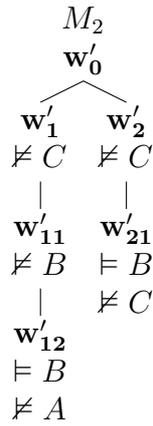
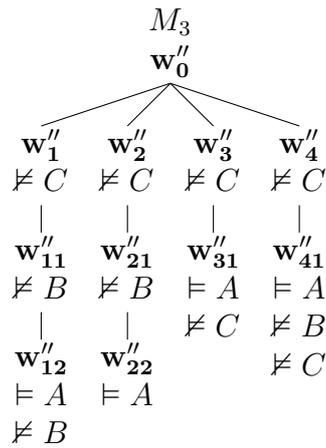
\begin{displaymath}    
{
  \Tree [.{$M_1$ \\ $\mathbf{w_{0}}$} 
                 [.{$\mathbf{w_{1}}$ \\ $\nvDash C$} [.{$\mathbf{w_{11}}$ \\ $\nvDash B$} {$\mathbf{w_{12}}$ \\ $\vDash A$} ]]
                 [.{$\mathbf{w_{2}}$ \\ $\nvDash C$} [.{$\mathbf{w_{21}}$ \\ $\nvDash B$} {$\mathbf{w_{22}}$ \\ $\vDash B$} ]]
                 [.{$\mathbf{w_{3}}$ \\ $\nvDash C$} {$\mathbf{w_{31}}$ \\ $\vDash A$ \\ $\nvDash B$} ]  
                 {$\mathbf{w_{4}}$ \\ $\nvDash C$}
                 {$\mathbf{w_{5}}$ \\ $\nvDash C$ \\ $\vDash B$}
        ]
}
\label{model:dummett-bifurcation-1}
\end{displaymath}

\begin{displaymath}    
{
  \Tree [.{$M_2$ \\ $\mathbf{w_{0}'}$} 
                 [.{$\mathbf{w_{1}'}$ \\ $\nvDash C$} [.{$\mathbf{w_{11}'}$ \\ $\nvDash B$} {$\mathbf{w_{12}'}$ \\ $\vDash B$ \\ $\nvDash A$} ]]
                 [.{$\mathbf{w_{2}'}$ \\ $\nvDash C$} {$\mathbf{w_{21}'}$ \\ $\vDash B$ \\ $\nvDash C$} ]  
        ]
}
\label{model:dummett-bifurcation-2}
\end{displaymath}

\begin{displaymath}    
{
  \Tree [.{$M_3$ \\ $\mathbf{w_{0}''}$} 
                 [.{$\mathbf{w_{1}''}$ \\ $\nvDash C$} [.{$\mathbf{w_{11}''}$ \\ $\nvDash B$} {$\mathbf{w_{12}''}$ \\ $\vDash A$ \\ $\nvDash B$} ]]
                 [.{$\mathbf{w_{2}''}$ \\ $\nvDash C$} [.{$\mathbf{w_{21}''}$ \\ $\nvDash B$} {$\mathbf{w_{22}''}$ \\ $\vDash A$} ]]
                 [.{$\mathbf{w_{3}''}$ \\ $\nvDash C$} {$\mathbf{w_{31}''}$ \\ $\vDash A$ \\ $\nvDash C$} ]
                 [.{$\mathbf{w_{4}''}$ \\ $\nvDash C$} {$\mathbf{w_{41}''}$ \\ $\vDash A$ \\ $\nvDash B$ \\ $\nvDash C$} ]                   
        ]
}
\label{model:dummett-bifurcation-3}
\end{displaymath}

Therefore, at the point in the tree where the rule \imply-left is applied to the
context $(\alpha_1, C)$ (the single labeled rule application in the
tree~\eqref{proof:counter-model-gen-dummet}) we have the join of these models.
Counter-model $M_4$ below represent this unification (to legibility we remove
repeated branches in $M_4$). Thus, $M_4 \nvDash \alpha$:

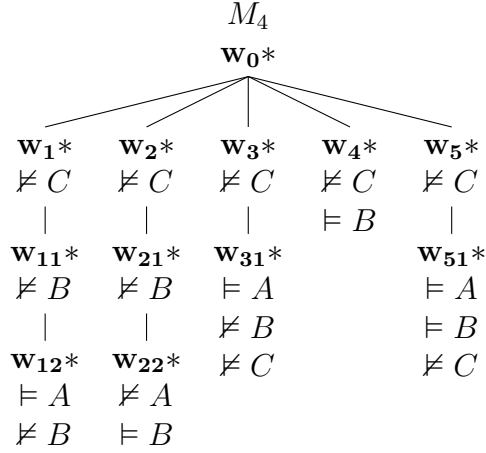
\begin{displaymath}    
{
  \Tree [.{$M_4$ \\ $\mathbf{w_{0}*}$} 
                 [.{$\mathbf{w_{1}*}$ \\ $\nvDash C$} [.{$\mathbf{w_{11}*}$ \\ $\nvDash B$} {$\mathbf{w_{12}*}$ \\ $\vDash A$ \\ $\nvDash B$} ]]
                 [.{$\mathbf{w_{2}*}$ \\ $\nvDash C$} [.{$\mathbf{w_{21}*}$ \\ $\nvDash B$} {$\mathbf{w_{22}*}$ \\ $\nvDash A$ \\ $\vDash B$} ]]
                 [.{$\mathbf{w_{3}*}$ \\ $\nvDash C$} {$\mathbf{w_{31}*}$ \\ $\vDash A$ \\ $\nvDash B$ \\ $\nvDash C$} ]  
                 {$\mathbf{w_{4}*}$ \\ $\nvDash C$ \\ $\vDash B$}
                 [.{$\mathbf{w_{5}*}$ \\ $\nvDash C$} {$\mathbf{w_{51}*}$ \\ $\vDash A$ \\ $\vDash B$ \\ $\nvDash C$} ]
        ]
}
\label{model:dummett-bifurcation-4}
\end{displaymath}

\end{example}

%%%%%%%%%%%%%%%%%%%%%%%%%%%%%%%%%%%%%%%%%%%%%%%%%%%%%%%%%%%%%%%%%%%%%%%%%%%%
\section{Conclusion}
\label{sec:conc}
%%%%%%%%%%%%%%%%%%%%%%%%%%%%%%%%%%%%%%%%%%%%%%%%%%%%%%%%%%%%%%%%%%%%%%%%%%%%
We presented here the definition of a sequent calculus for proof search in the
context of the Propositional Minimal Implicational Logic (\mil). Our calculus,
called \lmt, aims to perform the proof search for \mil formulas in a bottom-up,
forward-always approach. Termination of the proof search is achieved without
using loop checkers during the process. \lmt is a deterministic process, which
means that the system does not need an explicit backtracking mechanism to be
complete. In this sense, \lmt is a bicomplete process, generating Kripke
counter-models from search trees produced by unsuccess proving processes.

In the definition of the calculus, we also presented some translations between
deductive systems for \mil: ND to \ljarrow and \ljarrow to \lmt. We also
established a relation between \lmt and Fitting's Tableaux Systems for \mil
regards the counter-model generation in those systems.

We keep the development of a theorem prover for \lmt in
(\url{https://github.com/jeffsantos/GraphProver}).

As future work, we can enumerate some features to be developed or
extended in the system as well as some new research topics that can be
initialized.

\begin{itemize}
\item \textbf{Precise upper bound for termination} The upper bound used here for
  achieving termination in \lmt is a very high bound. Many non-theorems can be
  identified in a small number of steps. We can still explore options to shorten
  the size of the proof search tree. Even with theorems, our labeling
  mechanism, in conjunction with the usage of the restart rule, produces many
  repetitions in the proof tree.
  
\item \textbf{Compression and sharing} Following the techniques proposed by
  \cite{gordeev2016} we can explore new ways to shorten the size of proofs
  generated by \lmt.
  
\item \textbf{Minimal counter-models} The size of the generated counter-model in
  \lmt still takes into account every possible combination of subformulas,
  yielding Kripke models with quite a lot of worlds. There is still work to be
  done in order to produce smaller models. \cite{stoughton1996}~presents an
  implementation of the systems in~\cite{dyckhoff1992} and in~\cite{pinto1995}
  with the property of ``minimally sized, normal natural deduction proofs of the
  sequent, or it finds a "small" tree-based Kripke counter-model of the
  sequent'' using the words of the author. These references can be a good start
  point to improve~\lmt counter-model generation.
  
\end{itemize}
%%%%%%%%%%%%%%%%%%%%%%%%%%%%%%%%%%%%%%%%%%%%%%%%%%%%%%%%%%%%%%%%%%%%%%%%%%%%%%%%%%%

\bibliographystyle{apalike}
\bibliography{ref}

\end{document}